%% file: arxiv_version.tex
\documentclass[11pt]{article}  
\usepackage[margin=1in]{geometry}
\usepackage[utf8]{inputenc}
\usepackage[T1]{fontenc}
\usepackage{lmodern}
\usepackage{amssymb,amsmath,amsthm,amsfonts}
\usepackage{textgreek}
\usepackage{mathtools}
\usepackage{empheq} 
\usepackage{enumitem}
\usepackage[numbers,comma,sort&compress]{natbib}
\usepackage{authblk}
\usepackage{graphicx, caption, subcaption}
\usepackage{svg}
\usepackage{float}
\usepackage[ruled,vlined,linesnumbered]{algorithm2e}
\usepackage{algpseudocode}
\usepackage{siunitx}
\usepackage{physics}
\usepackage{footnote}
\usepackage{xcolor}
\usepackage{mathrsfs}
\usepackage{bbm}
\usepackage{bm}
\usepackage{braket}

\usepackage{hyperref}
\definecolor{darkgreen}{rgb}{0,0.5,0}
\hypersetup{
	unicode=false,          
	colorlinks=true,        
	linkcolor=red,          
	citecolor=darkgreen,    
	filecolor=magenta,      
	urlcolor=cyan           
}

\theoremstyle{plain}

\newtheorem{theorem}{Theorem}
\numberwithin{theorem}{section}

\newtheorem{lemma}[theorem]{Lemma}
\newtheorem{corollary}[theorem]{Corollary}
\newtheorem{claim}[theorem]{Claim}
\theoremstyle{definition}
\newtheorem{definition}[theorem]{Definition}

\newtheorem*{remark}{Remark}
\newtheorem{remarknumbered}[theorem]{Remark}

\newcommand{\eqn}[1]{\hyperref[eq:#1]{Eq.~\ref*{eq:#1}}}
\newcommand{\rem}[1]{\hyperref[rem:#1]{Remark~\ref*{rem:#1}}}
\newcommand{\clm}[1]{\hyperref[clm:#1]{Claim~\ref*{clm:#1}}}
\newcommand{\thm}[1]{\hyperref[thm:#1]{Theorem~\ref*{thm:#1}}}
\newcommand{\cor}[1]{\hyperref[cor:#1]{Corollary~\ref*{cor:#1}}}
\newcommand{\defn}[1]{\hyperref[def:#1]{Definition~\ref*{def:#1}}}
\newcommand{\assump}[1]{\hyperref[assump:#1]{Assumption~\ref*{assump:#1}}}
\newcommand{\lem}[1]{\hyperref[lem:#1]{Lemma~\ref*{lem:#1}}}
\newcommand{\prop}[1]{\hyperref[prop:#1]{Proposition~\ref*{prop:#1}}}
\newcommand{\fig}[1]{\hyperref[fig:#1]{Figure~\ref*{fig:#1}}}
\newcommand{\tab}[1]{\hyperref[tab:#1]{Table~\ref*{tab:#1}}}
\newcommand{\algo}[1]{\hyperref[algo:#1]{Algorithm~\ref*{algo:#1}}}
\renewcommand{\sec}[1]{\hyperref[sec:#1]{Section~\ref*{sec:#1}}}
\newcommand{\append}[1]{\hyperref[append:#1]{Appendix~\ref*{append:#1}}}
\newcommand{\fac}[1]{\hyperref[fac:#1]{Fact~\ref*{fac:#1}}}
\newcommand{\lin}[1]{\hyperref[lin:#1]{Line~\ref*{lin:#1}}}
\newcommand{\fnote}[1]{\hyperref[fnote:#1]{Footnote~\ref*{fnote:#1}}}

\def\>{\rangle}
\def\<{\langle}

\newcommand{\vect}[1]{\ensuremath{\boldsymbol{#1}}}

\renewcommand{\S}{\mathcal{S}}

\renewcommand{\d}{\mathrm{d}}

\renewcommand{\epsilon}{\varepsilon}

\newcommand{\bigO}[1]{
    \mathcal{O}\!\left({#1}\right)
}
\newcommand{\bigOdep}[2]{
    \mathcal{O}_{#1}\!\left({#2}\right)
}
\newcommand{\bigOmega}[1]{
    {\Omega}\!\left({#1}\right)
}

\newcommand{\sdg}{Schr\"{o}dinger}
\newcommand{\stoqma}{\normalfont \textsf{StoqMA}}
\newcommand{\bqp}{\normalfont \textsf{BQP}}
\newcommand{\bpp}{\normalfont \textsf{BPP}}
\newcommand{\pea}{\normalfont \textsf{P}}
\newcommand{\ma}{\normalfont \textsf{MA}}
\newcommand{\qma}{\normalfont \textsf{QMA}}
\newcommand{\np}{\normalfont \textsf{NP}}

\renewcommand\bra[1]{{\langle{#1}|}}
\makeatletter
\renewcommand\ket[1]{
  \@ifnextchar\bra{\k@t{#1}\!}{\k@t{#1}}
}
\newcommand\k@t[1]{{|{#1}\rangle}}
\makeatother

\numberwithin{equation}{section}

\input{preamble.tex}

\begin{document}

\title{On the Computational Complexity of \sdg\ Operators}
\author[1,2] {Yufan Zheng}
\author[2,3,4] {Jiaqi Leng}
\author[5,6] {Yizhou Liu}
\author[1,2,$\dagger$] {Xiaodi Wu}

\affil[1]{Department of Computer Science, University of Maryland}
\affil[2]{Joint Center for Quantum Information and Computer Science, University of Maryland}
\affil[3]{Department of Mathematics, University of Maryland}
\affil[4]{Simons Institute for the Theory of Computing and Department of Mathematics, UC Berkeley}
\affil[5]{Physics of Living Systems, Department of Physics, MIT}
\affil[6]{Department of Mechanical Engineering, MIT}
\affil[$\dagger$]{\href{mailto:xiaodiwu@umd.edu}{xiaodiwu@umd.edu}}

\date{}

\maketitle

\begin{abstract}
    We study computational problems related to the \sdg\ operator $H = -\Delta + V$ in the real space under the condition that (i) the potential function $V$ is smooth and has its value and derivative bounded within some polynomial of $n$ and (ii) $V$ only consists of $\bigO{1}$-body interactions.
    We prove that (i) simulating the dynamics generated by the \sdg\ operator implements universal quantum computation, i.e., it is $\bqp$-hard, and (ii) estimating the ground energy of the \sdg\ operator is as hard as estimating that of local Hamiltonians with no sign problem (a.k.a. stoquastic Hamiltonians), i.e., it is $\stoqma$-complete.
    This result is particularly intriguing because the ground energy problem for general bosonic Hamiltonians is known to be $\qma$-hard and it is widely believed that $\stoqma \varsubsetneq \qma$.
\end{abstract}

\section{Introduction}

Simulating quantum physical systems is one of the primary applications of quantum computers~\cite{feynman1982simulating,georgescu2014quantum,bauer2020quantum}.
The well-known \sdg\ equation describes the evolution of a quantum system,
$$
    i \frac{\dee}{\dee t} \ket{\Psi(t)} = H \ket{\Psi(t)},
$$
where $\ket{\Psi(t)}$ is the quantum wave function and $H$ is the system Hamiltonian.
For many quantum-mechanical systems, such as those of atoms, molecules, and solids~\cite{girvin2019modern,lewars2011computational}, $H$ takes the form of a \emph{\sdg\ operator}, 
$$
    H = - \Delta + V,
$$
where $\Delta$ is the Laplacian and $V\colon \mathbb{R}^n \to \mathbb{R}$ is a potential function with $n$ being the degree of freedom of the system. The real-space approach also appears in the investigation of quantum field theories~\cite{jordan2012quantum,preskill2019simulating}.

Due to their central role in quantum physics, the properties of \sdg\ operators have been extensively studied by the mathematics community since the 1950s~\cite{gel1951determination,reed1972methods,avron1978schrodinger,simon1982schrodinger,andrews2011proof,kato2013perturbation}. Recently, the theory of \sdg\ operators has found a wide range of applications in computer science, offering insights into classical and quantum optimization algorithms. Shi, Su, and Jordan~\cite{shi2023learning} highlighted the intriguing connections between \sdg\ operators and the stochastic gradient descent algorithm. \sdg\ operators have also emerged in the quantum computing literature as a powerful tool for designing novel quantum optimization algorithms~\cite{zhang2021quantum,liu2023quantum,leng2023qhd,leng2023quantum,chen2023quantum} that are capable to leverage the continuous structure of the objective function in the real space in the most natural way. 
This is in sharp contrast to traditional approaches like Quantum Adiabatic Algorithm or Quantum Approximate Optimization Algorithm which use qudit Hamiltonians~\cite{blekos2024review}.

Given the emerging importance of \sdg\ operators in computer science, it is thus natural to ask the following questions: 
\begin{center}
    \textit{What is the computational power of \sdg\ operators? Is it necessary to build a universal quantum computer to simulate them?}
\end{center}
At first glance, the answers to these questions may seem straightforward, as it is widely believed that general quantum-mechanical systems cannot be efficiently simulated using classical computers~\cite{aaronson2009quantum,a22bqp}. 
However, there exist efficient classical algorithms to simulate and/or compute certain physical properties of seemingly quantum processes, such as non-interacting fermions~\cite{terhal2002classical, knill2001fermionic,valiant2001quantum}, quantum circuit consisting of Clifford gates~\cite{gottesman1998heisenberg}, and ferromagnetic transverse field Ising model~\cite{bravyi2014monte}.
Until rigorous theoretical evidence is provided, we cannot trivially rule out the possibility of an efficient classical algorithm for simulating \sdg\ operators that does not harness the full power of universal quantum computing.

On the other hand, people have shown that many systems are indeed difficult to simulate via the notion of $\bqp$-completeness from computational complexity theory.
In brief, if simulating a system is $\bqp$-complete, then it is as hard as universal quantum computation.
Several quantum systems are shown to be $\bqp$-complete, including scattering in quantum field theory~\cite{jordan2018bqp}, Bose--Hubbard model~\cite{childs2013universal}, and almost all 2-local interactions on 2D square lattices~\cite{zhou2021strongly}. In other words, there is no efficient classical algorithm for any of these systems at all---unless there is one not only working for the specific system but \emph{all} possible quantum systems, i.e., $\bpp=\bqp$.

Besides the simulation of quantum dynamical processes, estimating the ground energy for a given quantum system (Hamiltonian) is another fundamental task in quantum computing, central to the field of Quantum Hamiltonian Complexity~\cite{gharibian2015quantum}.
In the past decades, people have identified the computational complexity of a variety of classes of Hamiltonians.
For many classes, we have strong evidence that they are intractable by results that pinpoint the \emph{exact} complexity classes they belong to.
An earliest example is the class of $k$-local Hamiltonians which can be written as a sum of terms that act on at most $k$ qubits.
Kitaev's seminal result shows that the task is $\qma$-complete for $5$-local Hamiltonians~\cite{kitaev2002classical}, where $\qma$ is the quantum analog of $\mathsf{NP}$~\cite{bookatz2014qma}.
Subsequent works improve this result by adding more and more restrictions: the task remains $\qma$-complete for general $2$-local Hamiltonians~\cite{kempe2006complexity}, $2$-local Hamiltonians on square lattices~\cite{oliveira2008complexity} and most $2$-local Hamiltonians on square lattices with a single fixed type of 2-qubit interactions~\cite{cubitt2018universal,zhou2021strongly};
If we do not restrict the topology, complexity classification results can be obtained for 2-local Hamiltonians~\cite{cubitt2016complexity,piddock2017complexity}.
Other classes of Hamiltonians have also been considered, such as Hamiltonians in one dimension~\cite{aharonov2009power,hallgren2013local}, Hamiltonians with no sign problem (a.k.a. stoquastic Hamiltonians)~\cite{bravyi2006complexity,bravyi2006merlin,bravyi2009complexity}, $\bigO{1}$-local Hamiltonians with succinct ground states~\cite{jiang2023local} or guiding states~\cite{gharibian2022dequantizing,weggemans2023guidable,cade2023improved}, Fermi-Hubbard~\cite{schuch2009computational} and Bose-Hubbard models~\cite{childs2014bose}.

\vspace{1em}
In this paper, we study the computational complexity of the system governed by a time-independent \sdg\ operator $H = -\Delta + V$.
Specifically, we are interested in the hardness of two tasks: (i) simulating the dynamics under $H$ and (ii) estimating the ground energy of $H$, as both tasks are fundamental in physics for a better understanding of corresponding systems~\cite{georgescu2014quantum,bauer2020quantum}.
For technical convenience, we restrict the domain of wave functions from $\mathbb{R}^n$ to a finite hypercube $[-1,1]^n$.
We also impose additional constraints on the regularity and complexity of the potential function $V(\vect{x})$, see \defn{sdgsim}. These constraints are necessary for the practical and physical relevance of our computational model. Technically, any restriction only \emph{strengthens} our computational hardness results.

\subsection{Overview of Results}
\paragraph{Simulating the \sdg\ operator is \bqp-hard.}
We show that the task of simulating the dynamics generated by the \sdg\ operator $H = -\Delta + V$ is \bqp-hard (\thm{main-bqp}) even if $V$ only consists of $2$-body interactions.
We also have strong evidence showing that this task itself is in \bqp, see~\thm{bqp-mem}.\footnote{The reason we did not obtain a \bqp-completeness result is mainly technical---we assume a Dirichlet boundary condition for wave functions but the simulation algorithm of~\cite{childs2022quantum} works under a periodic boundary condition.}
As a consequence, \emph{there is no efficient classical simulation algorithm} for \sdg\ operators, nor for the quantum algorithms involving the \sdg\ dynamics~\cite{zhang2021quantum,liu2023quantum,leng2023qhd,leng2023quantum,chen2023quantum}, unless $\bpp=\bqp$.

We prove the hardness result by reducing the simulation of the (time-independent) transverse field Ising model (TIM) to the task of simulating \sdg\ operators. 
As an interesting byproduct, we also demonstrated that simulating TIM is \bqp-complete (\thm{timbqp}), filling a conceptual gap in Hamiltonian complexity theory, as it has long been believed that TIM cannot be efficiently simulated using classical methods.
A perhaps more interesting interpretation is that implementing \emph{time-independent} TIM is powerful enough for universal quantum computing, which is particularly relevant to quantum utility given that there are quantum computers based on TIM architecture such as D-Wave~\cite{dwave2024}.

\paragraph{Determining the ground energy of the \sdg\ operator is \stoqma-complete.}
We show that for any potential $V$ that involves 2- or higher-order interactions, estimating the ground energy of the \sdg\ operator $H = -\Delta + V$ is \stoqma-complete (\thm{stoqma-mem} and \thm{main-stoqma}).
The complexity class \stoqma\ captures the computational hardness of finding the ground energy for local Hamiltonians with \emph{no sign problem}~\cite{bravyi2006merlin,bravyi2006complexity}, and can be placed between the classes \np\ and \qma. 
Unlike the simulation task, the ground energy problem of \sdg\ operators seems easier than that for general local Hamiltonians, a standard \qma-complete problem.\footnote{This is because $\stoqma \subseteq \mathsf{SBP} \subseteq \mathsf{AM}$~\cite{bravyi2006complexity,bravyi2006merlin}. If $\stoqma=\qma$ we will have an unlikely inclusion $\qma \subseteq \mathsf{AM}$.}

Though somewhat counterintuitive, this result does not contradict known physical principles: the \sdg\ operator essentially describes a \emph{separable} bosonic system, which lacks both the fermionic anti-commutation relation and non-trivial interactions between the kinetic and potential energies. In contrast, the ground energy problem of a general bosonic Hamiltonian is known to be \qma-hard~\cite{childs2014bose}.

\subsection{Technical Contributions}

The fact that the \sdg\ operator is defined in an infinite-dimensional Hilbert space and is unbounded (in the sense of operator $2$-norm) introduces several major technical difficulties, including (i) most existing quantum hardness results primarily rely on computational models with finite local degrees of freedom (e.g., qubit Hamiltonians), and (ii) many standard tools in Hamiltonian complexity theory, such as the Schrieffer--Wolff transformation~\cite{bravyi2011schrieffer}, no longer work in the infinite-dimensional setting. In this section, we briefly discuss our main technical contributions to address the aforementioned questions. These results may also be of independent interest in future research. 

\subsubsection{Encoding qubit states into bosonic modes}
The key step behind our $\bqp$- and $\stoqma$-hardness results is a perturbative reduction (\thm{main}) from transverse field Ising model (TIM) Hamiltonians to \sdg\ operators. 
To be more specific, we embed the whole spectrum of any TIM Hamiltonian $H$ into the low-energy subspace of a \sdg\ operator $\widehat{H}$ in an efficient way.
This embedding allows us to retain all physical properties of $H$, including spectral correspondence and simulation correspondence.\footnote{This resembles the definition of ``Hamiltonian simulation'' in the work of Cubitt, Montanaro and Piddock~\cite{cubitt2018universal}. But we cannot use their formal definition in this paper as theirs concerns reductions among finite dimensional systems consisting of qudits.}

It is worth noting that our constrcution is different from the well-known GKP code~\cite{gottesman2001encoding}. In the GKP code, the logcial codespace encoding $\ket{0}$ and $\ket{1}$ is the simultaneously $+1$ eigenspace of two displacement operators, while the encoding subspace we use in this work is spanned by the first two eigenstates of a \sdg\ operator with a non-quadratic potential field.

The high-level idea of our perturbative reduction is to map $\ket{+}$ and $\ket{-}$ to the ground state $\ket{\psi_0}$ and the first excited state $\ket{\psi_1}$ of a bosonic mode governed by some Hamiltonian $\widehat{X} \coloneqq -\frac{\dee^2}{\dee x^2} +\dwfunc$ where $\dwfunc$ is a symmetric double-well potential.
The correspondence allows us to embed $X$ in the TIM Hamiltonian to $c\widehat{X}$ for some $c$ because $\ket{\pm}$ can be seen as the ground state and the first excited state of $X$.
Furthermore, $\ket{0}$ (resp., $\ket{1}$) is mapped to $\ket{\logiczero} = \frac{1}{\sqrt{2}} (\ket{\psi_0} + \ket{\psi_1})$ (resp., $\ket{\logicone} = \frac{1}{\sqrt{2}} (\ket{\psi_0} - \ket{\psi_1})$).
Note that $\ket{\logicone}$ will be the reflection of $\ket{\logiczero}$ across the y-axis due to the symmetry of $\dwfunc$.
Choosing an appropriate $\dwfunc$ can make $\ket{\logiczero}$ (resp., $\ket{\logicone}$) a wave packet that is approximately confined in the region $x > 0$ (resp., $x < 0$);
See \fig{1}.
The geometric property of $\ket{\logiczero}$ and $\ket{\logicone}$ will thus make the embedding of a $Z_i Z_j$ term in the TIM Hamiltonian much easier: it is easy to imagine that (a smooth version of) the diagonal operator $\sgn(x_1 x_2)$ acting on $\ket{\logicvar{a}}\ket{\logicvar{b}}$ will approximately resemble $Z_1 Z_2$ acting on $\ket{a}\ket{b}$ for $a,b\in \{0,1\}$.

\begin{figure}[H]
    \centering
    \begin{subfigure}{0.4\textwidth}
        \includegraphics[width=\linewidth]{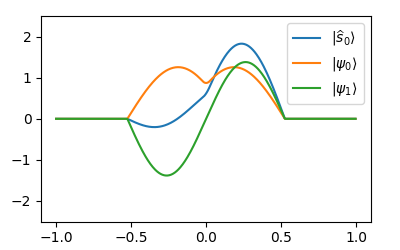}  
        \caption{$s = 100$}
    \end{subfigure}
    \begin{subfigure}{0.4\textwidth}
        \includegraphics[width=\linewidth]{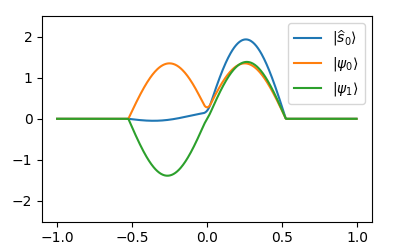}  
        \caption{$s = 500$}
    \end{subfigure}
    \caption{ First two eigenfunctions $\psi_0(x),\psi_1(x)$ and their linear combination $\logiczero(x) = (\psi_0(x)+\psi_1(x))/\sqrt{2}$ for the operator $\widehat{X} = -\frac{\dee^2}{\dee x^2} +\dwfunc$ with parameter $w = 0.05$ and $\ell = 0.5$, where $\dwfunc$ is defined in \sec{proof-outline}.}
    \label{fig:1}
\end{figure}

We believe that this idea can be generalized to embed other Hamiltonians into differential operators which are important in the fields of partial differential equations and functional analysis.

\subsubsection{Substitute of Schrieffer--Wolff theory}

When converting the idea in the previous section into an actual proof, a substitute of Schrieffer--Wolff theory (\lem{pert-sim}) is required to control the behavior of $e^{-iHt}$ when $H$ is approximately block-diagonal.
This is because the original Schrieffer--Wolff theory~\cite{bravyi2011schrieffer} works with a finite-dimensional Hilbert space while our case involves an infinite-dimensional Hilbert space.
Unlike the standard Schrieffer--Wolff theory, our \lem{pert-sim} does not depend on the spectral gap between two diagonal blocks of a Hamiltonian.
Therefore, our argument may be generalized and reused in future research if (i) the underlying Hilbert space is infinite-dimensional and/or (ii) the gap between two blocks cannot be easily controlled. 

\subsection{Organization}
The rest of the paper is organized as follows.
We introduce hardness results for TIM in \sec{tim}, which includes hardness of time-independent simulation and ground energy estimation.
\sec{algo} presents $\bqp$ and $\stoqma$ membership proofs for \sdg\ operators.
Finally, we prove the $\bqp$- and $\stoqma$-hardness results of \sdg\ operators in \sec{hardness}.

\section{Complexity of the transverse field Ising model} \label{sec:tim}
In this section we discuss the complexity of TIM Hamiltonians.
To be more concrete, \thm{tim} and \ref{thm:timbqp} are featured respectively for characterizing the complexity of (i) estimating the ground energy and (ii) time-independent simulation for TIM Hamiltonians which are defined as follows.
\begin{definition}(TIM Hamiltonian)
    An $n$-qubit \emph{TIM Hamiltonian} is that of a {\underline t}ransverse field {\underline I}sing {\underline m}odel taking the following form:
    $$
        H = \sum_{1 \leq u \leq n} \left(h_u X_u + g_u Z_u\right) + \sum_{1 \leq u,v \leq n} g_{u,v} Z_u Z_v,
    $$
    where $h_u,g_u,g_{u,v}$ are real coefficients and $X_u,Z_u$ are Pauli operators acting on the $u$th qubit.
\end{definition}
Let us introduce the definition of stoquastic Hamiltonian and complexity class \stoqma\ for later use.
\begin{definition}($k$-local termwise-stoquastic Hamiltonian)
    We call a Hamiltonian $H$ that acts on the $n$-qubit Hilbert space $\mathbb{C}^{2^n}$ \emph{$k$-local termwise-stoquastic} if it can be written as
    $$
        H = \sum_j H_j,
    $$
    where each $H_j$ acts nontrivially on at most $k$ qubits, and when written in the computational basis, $H_j$'s off-diagonal entries are nonpositive.
\end{definition}

\begin{definition}($\stoqma$~\cite{bravyi2006merlin}, informal)
    The complexity class \stoqma\ is an extension of the classical class $\ma$ where the verifier can accept quantum states as a proof instead of a bitstring in the definition of $\ma$. The verifier can then apply \emph{classical} reversible circuits on the state and measure the first qubit in the $\{\ket{+},\ket{-}\}$ basis.
    It accepts the proof iff the measurement outcome is `+'.
\end{definition}
The prefix ``Stoq'' in \stoqma\ comes from the fact that determining the ground energy of a $k$-local termwise-\underline{stoq}uastic Hamiltonian is a complete problem for \stoqma\ as long as $k \geq 2$~\cite{bravyi2006merlin,bravyi2006complexity}.
This resembles the fact that determining the ground energy of a general $k$-local Hamiltonian is $\qma$-complete for $k \geq 2$~\cite{kempe2006complexity}.

The following result from Brayvi and Hastings~\cite{bravyi2017complexity} characterizes the difficulty of estimating the ground energy for a general TIM Hamiltonian.

\begin{theorem}[\cite{bravyi2017complexity}] \label{thm:tim}
    Determining the ground energy of TIM Hamiltonian is $\stoqma$-complete.
    More specifically, estimating within error $1/\poly(n)$ the smallest eigenvalue of the TIM Hamiltonian
    $$
        H = \sum_{1 \leq u \leq n} (h_u X_u + g_u Z_u) + \sum_{1 \leq u,v \leq n} g_{u,v} Z_u Z_v,
    $$
    with $|h_u|,|g_u|,|g_{u,v}| \leq \poly(n)$ given explicitly, is complete for the complexity class \stoqma\ under \pea\ reductions.
\end{theorem}

Now let us turn to the complexity of simulating a TIM Hamiltonian in a time-independent fashion.
The problem is not settled in previous literature explicitly to our best knowledge.
However, it turns out that the simulation task is \bqp-complete by virtue of a few reductions and techniques from previous works.
\begin{definition}
    We call a family of Hamiltonians acting on $n$-qubit Hilbert spaces $\mathbb{C}^{2^n}$ \bqp-complete if the following problem is (\textsf{Promise})\bqp-complete under $\pea$ reductions:
    \begin{itemize}
        \item {\bf Input.} An integer $n$, a positive $t \leq \poly(n)$, the description of a Hamiltonian $H$ from the family, a product state $\ket{\psi} \in \mathbb{C}^{2^n}$ with each component being $\bigO{1}$-qubit and a projector $M = \ket{\mu}\bra{\mu} \otimes I$ where $\ket{\mu}$ is an $\bigO{1}$-qubit state.
        \item {\bf Output.} \texttt{Yes} if $\Tr{M\ket{\phi}\bra{\phi}} > 2/3$ or \texttt{No} if $\Tr{M\ket{\phi}\bra{\phi}} < 1/3$ given the promise that one must be the case, where $\ket{\phi} \coloneqq e^{-iHt}\ket{\psi}$.
    \end{itemize}
\end{definition}

\begin{theorem} \label{thm:timbqp}
    TIM Hamiltonians with coefficients bounded in $\poly(n)$ are $\bqp$-complete.
\end{theorem}
\begin{proof}
    It is obvious that it should be in $\bqp$ since general local Hamiltonian simulation is in $\bqp$.
    We focus on proving its $\bqp$-hardness below.
    By Theorem 42 of~\cite[SI]{cubitt2018universal} we know that $XX+ZZ$ Hamiltonians are \bqp-complete.\footnote{In fact, one can even restrict the topology of such Hamiltonians on a 2D square lattice and the \bqp-completeness will remain unchanged by virtue of~\cite{zhou2021strongly}. From that result, a similar $XX+YY$ family also finds its application in verification of quantum simulators~\cite{jackson2024accreditation}.}
    To be more specific, the family of Hamiltonians consist of $H$ taking the shape
    $$
        H = \sum_{1 \leq i,j \leq n} (a_{ij} X_i X_j + b_{ij} Z_i Z_j),
    $$
    where $\abs{a_{ij}},\abs{b_{ij}} \leq \poly(n)$.
    Now we employ the technique from~\cite{jordan2010quantum,janzing2006bqp} to embed the spectrum of $H$ into a 3-local termwise-stoquastic matrix (in the computational basis).
    Write $H = -\sum_k H_k$ where $H_k$ are 2-local real matrices of size $2^n$.
    Construct $\widetilde{H_k}$ of size $2^{n+1}$ from $H_k$ by the following replacement rule:
    $$
        x \mapsto \left\{
        \begin{array}{ll}
            \left( \begin{smallmatrix}x & 0\\ 0 & x\end{smallmatrix} \right), & x \geq 0, \\
            \left( \begin{smallmatrix}0 & -x\\ -x & 0\end{smallmatrix} \right) , & x < 0.
        \end{array}
        \right.
    $$
    $\widetilde{H_k}$ is 3-local if we interpret it as acting on $n+1$ qubits and it is clearly stoquastic.
    Note that $\widetilde{H_k} = H_k \otimes \ket{-}\bra{-} + \abs{H_k} \otimes \ket{+}\bra{+}$ where $\abs{\cdot}$ applies entrywise.
    Hence we can construct the following 3-local stoquastic Hamiltonian:
    $$
        \widetilde{H} \coloneqq - \sum_k \widetilde{H_k} = H \otimes \ket{-}\bra{-} + \left(- \sum_k \abs{H_k}\right) \otimes \ket{+} \bra{+}.
    $$
    Therefore, the evolution of $\ket{\psi}$ under Hamiltonian $H$ can be simulated by the evolution of $\ket{\psi}\ket{-}$ under $\widetilde{H}$.
    We conclude that 3-local termwise stoquastic Hamiltonians are also \bqp-complete.
    Finally, as detailed in Section 5.4 of~\cite[SI]{cubitt2018universal}, 3-local termwise-stoquastic Hamiltonians can be efficiently simulated by TIM Hamiltonians.
    This is achieved by chaining local simulation protocols from~\cite{bravyi2006complexity} and \cite{bravyi2017complexity}, which reduces 3-local termwise-stoquastic Hamiltonians to 2-local ones, and 2-local ones to TIM Hamiltonians, respectively.
\end{proof}

\section{Algorithms for \sdg\ Operators} \label{sec:algo}
From now on we will focus on the \sdg\ operator $H= -\Delta + V$ in the \emph{hypercube} $D$ with $V$ being a \emph{$\bigO{1}$-body potential}.
The hypercube constraint enforces a Dirichlet boundary condition on all wave functions such that they must vanish on the boundary of $D$.
The only exception is that in \sec{bqpmem} we will consider a periodic boundary condition instead of a Dirichlet one.
\begin{definition}(Hypercube)
    An $n$-dimensional hypercube is $D \coloneqq [-1,1]^n$.
    Throughout this paper we use $D$ to denote this specific hypercube unless noted otherwise.
\end{definition}
\begin{definition}($k$-body potential)
    A potential function $V: D \to \mathbb{R}$ is called \emph{$k$-body} if it only consists of $k$-body interaction:
    In other words, it can be written as the sum of functions that depends on at most $k$ coordinates.
    For example, the Coulomb potential of $n$ electrons on a straight line,
    $$
        V(r_1,\dots,r_n) = \sum_{1 \leq i < j \leq n} \frac{1}{\abs{r_i - r_j}},
    $$
    is a $2$-body potential.
\end{definition}
The main results in this section are \thm{stoqma-mem} and \thm{bqp-mem}, which addresses the upper bound of the complexity of (i) estimating the ground energy and (ii) time-independent simulation for \sdg\ operators, respectively.

\subsection{Estimating the Ground Energy is in StoqMA}

\begin{theorem} \label{thm:stoqma-mem}
    Determining the ground energy of \sdg\ operators with $\bigO{1}$-body potential is in \stoqma.
    More specifically, let $V: D \to \mathbb{R}$ be a $\bigO{1}$-body potential function with $C^1$-norm bounded within $\poly(n)$.
    Then, estimating within error $1/\poly(n)$ the smallest eigenvalue of the \sdg\ operator
    $
        H = -\Delta + V
    $
    is in the complexity class \stoqma.
\end{theorem}
\begin{proof}
    We will prove the theorem with a slightly generalized $D$ which is a hyperbox with its each dimension between $1/\poly(n)$ and $\poly(n)$.
    First, we use the vanilla finite difference method to discretize $H$.
    Write $D = D_1 \times D_2 \times \cdots \times D_n$.
    Fix a stepsize $\delta$.
    Let $\widetilde{D}$ be the discretized version of $D$: $\widetilde{D} = D \cap \{(j_1 \delta, j_2 \delta, \dots, j_n \delta): j_1,\dots,j_n \in \mathbb{Z}\}$.
    and similarly $\widetilde{D}_i \coloneqq D_i \cap \{j \delta: j \in \mathbb{Z}\}$.
    Obviously $\widetilde{D} \coloneqq \widetilde{D}_1 \times \cdots \times \widetilde{D}_n$.
    Suppose $V = \sum_j V_j$ is the decomposition of the $\bigO{1}$-body potential $V$.
    Define $\widetilde{V}_j \coloneqq \sum_{\vect{x} \in \widetilde{D}} V_j(\vect{x})\ket{\vect{x}} \bra{\vect{x}}$.
    Define $\widetilde{T}_j$ to be a a tri-diagonal matrix that discretizes $-\frac{\partial^2}{\partial x_j^2}$: 
    $$\widetilde{T}_j \coloneqq \delta^{-1} \left( I - \sum_{\substack{x,y \in \widetilde{D}_j \\ \abs{x-y}=\delta}} \ket{x}\bra{y} \right).$$
    Let
    $$
        \widetilde{H} \coloneqq \sum_{1 \leq j \leq n} \widetilde{T}_j + \sum_j \widetilde{V}_j.
    $$
    Now by the standard analysis of finite element methods~\cite{babuska1991eigenvalue} we have $\abs{\mineigen(H) - \mineigen(\widetilde{H})} \leq \delta^{\Omega(1)}$.
    
    Next, we invoke the unary embedding method from~\cite{leng2024expanding} for $\widetilde{H}$.
    Intuitively, we will embed the Hilbert space of a qudit $\mathbb{C}^{m+1}$ into that of $m$ qubits $\mathbb{C}^{2^{m}}$ by the correspondence $U: \ket{j} \mapsto \ket{0^j 1^{m-j}}$ for $j=0,1,\dots,m$.
    Define $\mathcal{S}$ to be the embedding subspace $\spn\{\ket{0^j 1^{m-j}}\}_j$.
    For any $(m+1)$-dimensional real Hermitian matrix $A = \sum_{0 \leq i,j \leq m} \alpha_{i,j} \ket{i} \bra{j}$ we define its embedding to be
    $$
        \sigma(A) = \left(\alpha_{0,0} I_0 + \frac12 \sum_{0\le j \le m-1} (\alpha_{j-1,j-1}-\alpha_{j,j}) (I_j - Z_j)\right) + \sum_{0\le j < k \le m} X_{k-1}\otimes \dots \otimes X_{j+1} \otimes \alpha_{j,k} X_j.
    $$
    It is straightforward to verify that $\sigma(A) = UAU^\dagger$ in $\mathcal{S}$~\cite{leng2024expanding}.
    Furthermore, we have:
    \begin{itemize}
        \item $\sigma(I)$ is the identity.
        \item $\sigma(A)$ is 1-local if $A$ is tri-diagonal.
        \item $\sigma(A)$ is diagonal (resp., stoquastic) if $A$ is diagonal (resp., stoquastic).
    \end{itemize}
    
    Now, we want to embed the discretized \sdg\ operator $\widetilde{H} = \sum \widetilde{T}_j + \sum \widetilde{V}_j$.
    Define
    $$
        \widehat{T_j} \coloneqq \sigma(I)^{\otimes j} \otimes \sigma(\widetilde{T}) \otimes \sigma(I)^{\otimes n-j-1} \text{ and } \widehat{V_j}\coloneqq\sigma(\widetilde{V}_j) \otimes \sigma(I)^{\otimes n-k}
    $$
    assuming $\widetilde{V_j}$ acts nontrivally on the first $k$ qudits.
    This naturally generalizes to a definition for general $\widehat{V_j}$.
    Let $\widehat{H} \coloneqq \sum \widehat{T}_j + \sum \widehat{V}_j$.
    By our construction, $\widehat{H} = U^{\otimes n} \widetilde{H} U^{\dagger \otimes n}$ in $\mathcal{S}^{\otimes n}$ and $\widehat{H} \mathcal{S}^{\otimes n} \subseteq \mathcal{S}^{\otimes n}$. 

    Finally, we introduce a ``penalty'' diagonal Hamiltonian $\widetilde{Q}$ that separates its lowest eigenspace $\mathcal{S}^{\otimes n}$ from the remaining eigenspaces:
    $$
        \widehat{Q} \coloneqq \sum_{1 \leq i \leq n} I^{\otimes i-1} \otimes \left( 2I -\sum_{0 \leq j \leq m-2} Z_j Z_{j+1} + Z_0 - Z_{m-1} \right) \otimes I^{n-i}.
    $$
    It can be verify that its spectral gap is $\eigen_1(\widehat{Q}) - \eigen_0(\widehat{Q}) = 4$ and $\eigen_0(\widehat{Q}) = \mineigen(\widehat{Q}) = 0$.
    Now we consider the spectrum of the operator
    $
        \widehat{H}_{\star} \coloneqq \widehat{H} + c\widehat{Q}
    $.
    To analyze the smallest eigenvalue of $\widehat{H}_{\star}$, note that it is block-diagonal regarding subspaces $\mathcal{S}^{\otimes n}$ and $(\mathcal{S}^{\otimes n})^\perp$ because (i) $\widehat{Q}$ has the minimal eigenvalue $0$ with the eigenspace $\mathcal{S}^{\otimes n}$ and (ii)  $\widehat{H} \mathcal{S}^{\otimes n} \subseteq \mathcal{S}^{\otimes n}$.
    The smallest eigenvalue of $\widehat{H}_{\star}$ is just that of $\widetilde{H}$ because $\widehat{H}_{\star} = \widehat{H} + c \cdot 0 = U^{\otimes n} \widetilde{H} U^{\dagger \otimes n}$ when restricted on the subspace $\mathcal{S}^{\otimes n}$.
    On the other hand, for any $\ket{\psi} \in (\mathcal{S}^{\otimes n})^\perp$ we have
    $$
        \braket{\psi | \widehat{H}_{\star} | \psi} = \braket{\psi | \widehat{H} | \psi} + c \braket{\psi | \widehat{Q} | \psi} \geq -\norm{\hat{H}} + c \cdot 4 \geq -p(n) + 4c.
    $$
    where $p(n)$ is some polynomial in $n$ by the naive estimation of $\hat{H}$.
    Choosing $c = p(n)$ we are guarantee that $\braket{\psi | \widehat{H}_{\star} | \psi} > p(n) > \mineigen(\widetilde{H})$.
    Hence the smallest eigenvalue of  $\widehat{H}_{\star}$ on subspace $(\mathcal{S}^{\otimes n})^\perp$ is greater than that on $\mathcal{S}^{\otimes n}$.
    Therefore, $\mineigen(\widehat{H}_{\star}) = \mineigen(\widetilde{H})$.
    To conclude the proof simply note that $\widehat{H}_{\star}$ is 2-local termwise-stoquastic by construction and we can choose some enough small $\delta \geq 1/\poly(n)$ such that the reduction remains efficient.
\end{proof}

\subsection{Time-Independent Simulation is in BQP} \label{sec:bqpmem}

To show that time-independent of \sdg\ equation with periodic boundary condition is in \bqp, we simply invoke the \emph{pseudospectral algorithm} of Childs et al.~\cite[Theorem 8]{childs2022quantum} with a refined analysis~\cite[Theorem 4]{leng2023quantum}.
We first recall the notion of \emph{mesh} and \emph{Fourier series}.
\begin{definition}[Mesh]
    An $n$-dimensional mesh of order $m$ is the discretization of the hypercube $[0,1)^n$, written as $\Gamma_m^n \coloneqq \{0,\frac{1}{2m+1},\dots,\frac{2m}{2m+1}\}^n$. 
\end{definition}
\begin{definition}[Fourier series]
    Let $\vect{x} = (x_1,\dots,x_n) \in \mathbb{R}^n$.
    Write $\vect{k} = (k_1,\dots,k_n) \in \mathbb{Z}^n$.
    A \emph{Fourier series of order $m$} is of the form
    $$
        f(\vect{x})=\sum_{-m \leq k_1,\dots,k_n \leq m} c_{\vect{k}} \prod_{1 \leq j \leq n} \exp(2 \pi i k_j x_j).
    $$
    Note that $f(\vect{x})$ is periodic with translation lattice $\mathbb{Z}^n$.
    Furthermore, $f$ is uniquely determined by its value on the mesh $\Gamma_m^n$.
\end{definition}
\begin{definition}[Pseudospectral algorithm~\cite{childs2022quantum}]
    The purpose of this algorithm is to simulate the time-depedent \sdg\ equation
    $$
        i \frac{\partial}{\partial t}\ket{\Psi(\vect{x},t)} = \left(-\Delta + V(\vect{x},t)\right)\ket{\Psi(\vect{x},t)}
    $$
    from $t=0$ to $t=T$ over the hypercube $D=[-1,1]^n$ with a periodic boundary condition.
    The algorithm maintains the approximation of $\Psi(\vect{x},t)$, denoted as $\widetilde{\Psi}(\vect{x},t)$, on the mesh $\Gamma_m^n$ in the form of a quantum state 
    $$
        \ket{\widetilde{\Psi}(t)} \coloneqq \sum_{\vect{x} \in \Gamma_m^n} \widetilde{\Psi}(\vect{x},t) \ket{\vect{x}}.
    $$
    \begin{itemize}
        \item {\bf Input.} Dimension $n$; Truncation number $m$; Evolution time $T$; The \emph{quantum} input state $\ket{\widetilde{\Psi}(0)}$ such that ${\widetilde{\Psi}(\cdot,0)}={{\Psi}(\cdot,0)}$; An oracle $O_V$ which returns the value of $V(\vect{x},t)$ given $\vect{x}$ and $t$; Two metadata oracles $O_{\rm inv}$ and $O_{\rm norm}$ (See~\cite[Sec. 3.2]{childs2022quantum}) which act trivially if $V$ is independent on $t$.
        \item {\bf Output.} A state $\ket{\widetilde{\Psi}(T)}$.
    \end{itemize}    
\end{definition}

\begin{theorem}[\cite{childs2022quantum}] \label{thm:spectral-method}
    Suppose $V(\vect{x},t)\colon D \times [0, T] \to \mathbb{R}$ is bounded, smooth in $\vect{x}$ and $t$, and periodic in $\vect{x}$. Moreover, we assume that the initial wave function $\Psi(\vect{x},0)$ is periodic on $D$ and $V$ is $L$-Lipschitz in $t$. Define $\|V\|_{\infty,1}\coloneqq \int^{T}_{0}\|V(\cdot,t)\|_{\infty}~\d t$.
    Then, choosing $m \gg \log \eta$ can guarantee that 
    $$
        \abs{\widetilde{\Psi}(\vect{x},T) - \Psi(\vect{x},T)} \leq \eta
    $$
    for every $\vect{x} \in \Gamma_m^n$.
    The cost of the algorithm includes the following:
    \begin{enumerate}
        \item Queries to the quantum evaluation oracle $O_V$, $O_{\rm inv}$ and $O_{\rm norm}$: $\bigO{\|V\|_{\infty,1} \frac{\log(\|V\|_{\infty,1}/\eta)}{\log\log(\|V\|_{\infty,1}/\eta)}}$,
        \item 1- and 2-qubit gates: 
        $$\bigO{\|V\|_{\infty,1}\left(\polylog\left(T/\eta\right)+\log^{2.5}\left(L\|V\|_{\infty,1}/\eta\right) + d\log m\right)\frac{\log(\|V\|_{\infty,1}/\eta)}{\log\log(\|V\|_{\infty,1}/\eta)}}.$$
    \end{enumerate}
\end{theorem}
\begin{remark}
    Compared to the original statement in~\cite[Theorem 4]{leng2023quantum}, we made the dependence on $m$ explicit and replaced the $\poly(z)$ by a $\polylog(T/\eta)$ in \thm{spectral-method} where $V$ is able to be computed with $z$ bits of precision.
    This can be justified as follows.
    Setting $z \coloneqq c\log(T/\eta)$ is sufficient for obtaining the value of $V$ with error at most $\bigO{c\eta/T}$.
    By Duhamel's principle we have
    \begin{align*}
        &\phantom{=}\ \exp_{\mathcal{T}}(\int_0^t -iA(s) \dee s) - \exp_{\mathcal{T}}(\int_0^t -iB(s) \dee s) \\
        &= -i \int_0^t \left( \exp_{\mathcal{T}}\left(\int_s^t -iA(\tau) \dee \tau\right) (A(s) - B(s)) \exp_{\mathcal{T}}\left(\int_0^s -iB(\tau) \dee \tau\right) \right) \dee s.
    \end{align*}
    Let $A(t)$ be the ideal \sdg\ operator $-\Delta + V(\vect{x},t)$ and $B(t)$ be one with $V$ truncated with error $\bigO{c\eta/T}$ as discussed above.
    We immediately obtain that the truncation incurs at most $\bigO{c\eta/T} \cdot t = \bigO{c\eta}$ error in $\Psi(\vect{x},t)$.
\end{remark}

Let us formally define the simulation problem before proceeding to a proof.

\begin{definition}[\sdgsim] \label{def:sdgsim}
    \sdgsim\ is a decision problem with its input and output listed as below.
    \begin{itemize}
        \item {\bf Input.} Dimension $n$; Evolution time $t \leq \poly(n)$; A classical $\poly(n)$-time explicit oracle to compute $V(\vect{x})$ where $V: D \to \mathbb{R}$ is a smooth potential function with its $C^1$-norm bounded within $\poly(n)$\footnote{i.e., the function has supremum norm $\poly(n)$ and it is $\poly(n)$-Lipschitz.}; $\bigO{n}$ many of classical $\poly(n)$-time explicit oracles to compute $\bigO{1}$-dimensional smooth wave functions $\psi_j(\vect{x})$ and $\mu(\vect{x})$ with their $C^1$-norms bounded within $\poly(n)$.
        \item {\bf Output.} \texttt{Yes} if $\Tr{M \ket{\phi} \bra{\phi}} > 2/3$ or \texttt{No} if $\Tr{M \ket{\phi} \bra{\phi}} < 1/3$ given the promise that one must be the case, where
        $$
            \ket{\phi} \coloneqq e^{-i (-\Delta + V) t} \ket{\psi} \coloneqq e^{-i (-\Delta + V) t} \left( \bigotimes_j \ket{\psi_j} \right) \textrm{ and } M \coloneqq  \ket{\mu}\bra{\mu} \otimes I.
        $$
    \end{itemize} 
    We will have two variations of the problem depending on which boundary condition of wave functions we are considering: the \emph{Dirichlet} boundary condition and the \emph{periodic} boundary condition.
\end{definition}

\begin{theorem} \label{thm:bqp-mem}
    $\sdgsim$ with periodic boundary condition is in \bqp.
\end{theorem}
\begin{proof}
    Without loss of generality we can assume $\ell=1$ since a time dilation argument can scale $D$ by an arbitrary $\poly(n)$ factor.
    We describe the polynomial-time quantum algorithm below.
    Let the truncation parameter $m$ be determined later.
    We prepare the state $\ket{\widetilde{\Psi}(0)}$ in the pseudospectral algorithm as the approximation of $\bigotimes_j \ket{\psi_j}$ on the mesh $\Gamma_m^n$:
    $$
        \ket{\widetilde{\Psi}(0)} \coloneqq \bigotimes_j \ket{\widetilde{\psi}_j} \coloneqq \bigotimes_j \left( \sum_{\vect{x} \in \Gamma_m^{n_j}} \psi_j(\vect{x}) \ket{\vect{x}} \right),
    $$
    where $n_j$ is the dimension of wave function $\psi_j(\cdot)$.
    This can be done via standard quantum black-box state preparation with $\bigO{nm^{\bigO{1}}}$ oracle calls to values of $\psi_j(\cdot)$~\cite{grover2000synthesis}.
    Similarly, we construct the circuit $U$ such that $U\ket{\texttt{0}} = \ket{\widetilde{\mu}}$,
    where $\ket{\widetilde{\mu}}$ is the mesh approximation to $\ket{{\mu}}$ just like $\ket{\widetilde{\psi}_j}$ to $\ket{{\psi}_j}$.
    The implementation of $U$ uses at most $\bigO{m^{\bigO{1}}}$ oracle calls to values of $\mu(\cdot)$.
    We run the pseudospectral algorithm to obtain the final state $\ket{\widetilde{\Psi}(t)}$ by feeding it the explicit oracle to $V(\cdot)$.
    The runtime is $\bigO{t \polylog (t) \poly(n) \log m} = \bigO{\poly(n) \log m}$ if we set the error $\eta$ to be $e^{-n}$.
    It is obvious that the probability of measuring $U^{-1} \ket{\widetilde{\Psi}(t)}$ in the computational basis and then getting result \texttt{0} is exactly $\Tr{\widetilde{M} \ket{\widetilde{\Psi}(t)} \bra{\widetilde{\Psi}(t)}}$, where $\widetilde{M} = \ket{\widetilde{\mu}}\bra{\widetilde{\mu}}$.
    Now, $\ket{\widetilde{\Psi}(t)}$ is the approximation of $\ket{{\Psi}(t)} = e^{-i(-\Delta + V)t}\ket{\widetilde{\Psi}(0)}$ with pointwise error at most $e^{-n}$, and $\abs{\braket{\widetilde{\mu} | \widetilde{\Psi}(0)} - \braket{{\mu} | \psi}} \leq \bigO{\sqrt{n}L^2/m}$ if $\psi(\cdot) = \prod_j \psi_j(\cdot)$ and $\mu(\cdot)$ are $L$-Lipschitz.
    Therefore,
    $$
        \abs{\Tr{\widetilde{M} \ket{\widetilde{\Psi}(t)} \bra{\widetilde{\Psi}(t)}} - \Tr{M \ket{\phi} \bra{\phi}}} \leq \bigO{\frac{\sqrt{n}L^2}{m}+e^{-n}}.
    $$
    We conclude the proof by that $L \leq \poly(n)$ from the $C^1$-norm conditions of $\psi_j(\cdot)$ and $\mu(\cdot)$, and hence some $m \leq \poly(n)$ will make the error above smaller than $1/100$ when $n \to \infty$.
\end{proof}

\section{Hardness of \sdg\ operators} \label{sec:hardness}

The two main theorems of this section, \thm{main-stoqma} and \ref{thm:main-bqp}, establish the hardness of \sdg\ operators in terms of the complexity of finding the ground energy and time-independent simulation.
Their correctness follow from (i) \thm{main}, a perturbative reduction from \sdg\ operators to TIM Hamiltonians and (ii) the hardness results of TIM Hamiltonians, namely \thm{tim} and \ref{thm:timbqp}.

The rest of this section is organized as follows:
In \sec{proof-outline} we provide the essential idea behind \thm{main} by motivating a perturbative reduction in the 1-dimensional case. 
In \sec{proof-of-lem-main} we present the proof of \lem{main}, which serves as our main lemma that plays the central role in showing the correctness of the main perturbative reduction (\thm{main}).
We need some preparation before that:
The first ingredient is the semiclassical analysis of \sdg\ operators with $V$ being double-well functions in \sec{semiclassical}.
Then we introduce some perturbation theory tools in \sec{pert}.

\begin{theorem}[\stoqma-hardness] \label{thm:main-stoqma}
    Determining the ground energy of \sdg\ operators with 2-body potential is \stoqma-hard.
    More specifically, let $V: [-1,1]^n \to \mathbb{R}$ be a 2-body potential function with $C^1$-norm bounded within $\poly(n)$.
    Then, estimating within error $1/\poly(n)$ the smallest eigenvalue of the \sdg\ operator $H = \Delta + V$ is \stoqma-hard.
\end{theorem}

\begin{theorem}[\bqp-hardness] \label{thm:main-bqp}
    \sdgsim\ with Dirichlet boundary condition is $\bqp$-hard even if $V$ is 2-body.
\end{theorem}

\begin{theorem}[Perturbative reduction] \label{thm:main}
    Let $\epsilon \coloneqq 1/P(n)$ and $t < P(n)$ for some polynomial $P(\cdot)$.
    For any TIM Hamiltonian 
    $$
        H = \sum_{1 \leq u \leq n} (a_u X_u + b_u Z_u) + \sum_{1 \leq u,v \leq n} b_{u,v} Z_u Z_v
    $$
    such that $|a_u|,|b_u|,|b_{u,v}| \leq P(n)$, 
    there exists a \sdg\ operator
    $$
        \widehat{H} = - \Delta + V
    $$
    and a $\poly(n)$-time computable $C \leq \poly(n)$ satisfying the following:
    \begin{enumerate}[label=(\roman*)]
        \item The potential $V: [-1,1]^n \to \mathbb{R}$ is $2$-body, smooth, $\poly(n)$-time computable and has $C^1$-norm bounded within $\poly(n)$.
        \item $\abs{C \eigen_k(\widehat{H}) - \eigen_k(H)} \leq \epsilon$ for every $k \in \{0,1,\dots,2^n-1\}$.
        \item $\norm{ W^\dagger P_{\mathcal{S}} e^{-iC\widehat{H}t}W - e^{-iHt}} \leq \epsilon$, where $W \coloneqq \bigotimes_{1 \leq k \leq n} (\ket{\psi_0^{k}}\bra{-} + \ket{\psi_1^{k}}\bra{+})$ shifts basis and $\mathcal{S}$ is the image of $W$.  $\ket{\psi_0^k}$ and $\ket{\psi_1^k}$ are $1$-dimensional.
        Moreover, wave functions $\psi_j^k(\cdot)$ are smooth, $\poly(n)$-time computable and have $C^1$-norm bounded within $\polylog(n)$.
    \end{enumerate}
\end{theorem}
\begin{proof}
    Without loss of generality assume $a_u \geq 0$ because otherwise one can conjugate $H$ with $Z_u$ and this conjugation simply flip the $u$th qubit in the $\{\ket{+},\ket{-}\}$ basis.
    Now compare the theorem statement of \thm{main} and that of \lem{main}: the only obstacle from directly invoking \lem{main} is that it has an additional assumption $a_u \geq 1/M_2$ for every $u$.
    The assumption can be safely removed by \lem{add-small-number-to-a} which uses a perturbative reduction.
\end{proof}

\begin{lemma} \label{lem:main}
    Fix $M_1,M_2,\epsilon_1 > 0$ and $t \geq 0$.
    For any TIM Hamiltonian 
    $$
        H = \sum_{1 \leq u \leq n} (a_u X_u + b_u Z_u) + \sum_{1 \leq u,v \leq n} b_{u,v} Z_u Z_v
    $$
    such that $|a_u|,|b_u|,|b_{u,v}| \leq M_1$ and $a_u > 1/M_2$, 
    there exists a \sdg\ operator
    $$
        \widehat{H} = - G^\star \Delta + V
    $$
    and a polynomial $P \leq \poly(n,M_1,M_2,1/\epsilon_1,t)$ satisfying the following:
    \begin{enumerate}[label=(\roman*)]
        \item $1 \leq G^\star \leq P$ is $\poly(n)$-time computable and the potential $V: [-1,1]^n \to \mathbb{R}$ is $2$-body, smooth, $\poly(n)$-time computable and has $C^1$-norm bounded within $P$.
        \item $\abs{\eigen_k(\widehat{H}) - \eigen_k(H)} \leq \epsilon_1$ for every $k \in \{0,1,\dots,2^n-1\}$.
        \item $\norm{ W^\dagger P_{\mathcal{S}} e^{-i\widehat{H}t}W - e^{-iHt}} \leq \epsilon_1$, where $W \coloneqq \bigotimes_{1 \leq k \leq n} (\ket{\psi_0^{k}}\bra{-} + \ket{\psi_1^{k}}\bra{+})$ shifts basis and $\mathcal{S}$ is the image of $W$.  $\ket{\psi_0^k}$ and $\ket{\psi_1^k}$ are $1$-dimensional.
        Moreover, wave functions $\psi_j^k(\cdot)$ are smooth, $\poly(n)$-time computable and have $C^1$-norm bounded within $\polylog(P)$.
    \end{enumerate}
\end{lemma}

\begin{lemma} \label{lem:add-small-number-to-a}
    Fix $M_1,M_2 > 0$.
    For any TIM Hamiltonian 
    $$
        H = \sum_{1 \leq u \leq n} (a_u X_u + b_u Z_u) + \sum_{1 \leq u,v \leq n} b_{u,v} Z_u Z_v
    $$
    such that $\abs{a_u},\abs{b_u},\abs{b_{u,v}}\leq M_1$ and $a_u \geq 0$,
    there exists another TIM Hamiltonian
    $$
        \widetilde{H} = \sum_{1 \leq u \leq n} (\widetilde{a}_u X_u + \widetilde{b}_u Z_u) + \sum_{1 \leq u,v \leq n} \widetilde{b}_{u,v} Z_u Z_v
    $$
    with all coefficients being $\poly(n)$-time computable, which satisfies the following:
    \begin{enumerate}[label=(\roman*)]
        \item $\abs{\widetilde{a}_u},\abs{\widetilde{b}_u},\abs{\widetilde{b}_{u,v}}\leq \max\{M_1,1/M_2\}$ and $\widetilde{a}_u \geq 1/M_2$ for any $u$.
        \item $\abs{\eigen_k(\widetilde{H}) - \eigen_k(H)} \leq n/M_2$ for every $k \in \{0,1,\dots,2^n-1\}$.
        \item $\norm{ e^{-i\widetilde{H}t} - e^{-iHt}} \leq nt/M_2$.
    \end{enumerate}
\end{lemma}
\begin{proof}
    We simply choose $\widetilde{a}_u \coloneqq \min\{a_u,1/M_2\}$, $\widetilde{b}_u \coloneqq b_u$ and $\widetilde{b}_{u,v} \coloneqq b_{u,v}$.
    Note that $\abs{\widetilde{a}_u - a_u} \leq 1/M_2$.
    Therefore $\norm{H - \widetilde{H}} \leq n/M_2$ and hence (ii) and (iii) follow by Weyl's inequality (\lem{weyl}) and \lem{pert-sim-diff}, respectively.
\end{proof}

\subsection{Proof Idea} \label{sec:proof-outline}

We explain the rough idea with a 1-dimensional toy case (i.e., embedding a qubit into a single bosonic mode).
Suppose we want to simulate the TIM Hamiltonian $H = aX + bZ$ with some \sdg\ operator $\widehat{H}$.
A natural approach will be to first assume that the reduction map $\sigma: H \mapsto \widehat{H}$ is \emph{linear} for a ``reasonable'' range of $a$ and $b$.
Let us first consider what $\widehat{Z} \coloneqq \sigma(Z)$ to choose:
Obviously it should be the linear combination of $\frac{\dee^2}{\dee x^2}$ and a diagonal potential function $f$ since these are only two allowed terms in a \sdg\ operator $\hat{H}$.
Now, many choices are possible but there is one promising candidate being the sign function $\hat{Z} = \sgn(x)$ with two advantages: (i) this choice easily generalizes to higher dimensions and (ii) for any wave function $\psi(x)$ staying in $x>0$ (resp., $x<0$) we have $\hat{Z} \ket{\psi} = \ket{\psi}$ (resp. $\hat{Z} \ket{\psi} = -\ket{\psi}$), which resembles the fact $Z\ket{0}=\ket{0}$ (resp.,  $Z\ket{1}=-\ket{1}$).
Sticking with this choice, the next step is to find an $\widehat{X} \coloneqq \sigma(X)$ that is compatible with the choice of $\widehat{Z}$ in the sense that we can find wave functions $\logiczero(x)$ and $\logicone(x)$ such that
\begin{align*}
    \widehat{X}(\alpha \ket{\logiczero} + \beta \ket{\logicone}) \approx \alpha \ket{\logicone} + \beta \ket{\logiczero}, \\
    \widehat{Z}(\alpha \ket{\logiczero} + \beta \ket{\logicone}) \approx \alpha \ket{\logiczero} - \beta \ket{\logicone}.
\end{align*}
In other words, $\widehat{Z}, \widehat{X}, \ket{\logiczero},\ket{\logicone}$ approximately retain the algebraic relation among $Z, X, \ket{0}, \ket{1}$.
We already know that any $\logiczero(x)$ (resp., $\logicone(x)$) that concentrates in the region $x<0$ (resp., $x>0$) will fit the bill for $\widehat{Z}$ from the previous discussion about $\widehat{Z}$.
Therefore we assume that they have this concentration property from now on.
Alas, it turns out that we cannot find a compatible $\widehat{X}$ unless the error hidden in $\approx$ can be a constant.\footnote{But we need the error to be at least smaller than $1/n$ because error adds up when generalizing from 1 dimension to $n$ dimensions.}
However, note that what we actually wanted to do is maintaining the \emph{spectral} property from $H$ to $\hat{H}$.
Therefore an $\hat{X}$ which acts on $\ket{\logiczero},\ket{\logicone}$ like $c_1X+c_2I$ acting on $\ket{0},\ket{1}$ for some $c_1,c_2$ will suffice.
Another way to see this is that time-independent simulation under Hamiltonian $c_1H+c_2I$ is equivalent to that under $H$ modulo phase difference and time dilation.
In light of this, we can heuristically find $\widehat{X}$ by considering the spectral property of $c_1X+c_2I$: its eigenstates are $\ket{0}+\ket{1}$ and $\ket{0} - \ket{1}$.
This means a valid $\widehat{X}$ should have eigenfunctions that look like $\logiczero(x)+\logicone(x)$ and $\logiczero(x)-\logicone(x)$.
Since we want to embed $H$ to the low-energy subspace of $\widehat{H}$ we prefer these two eigenstates being the ground state and the first excited state of $\widehat{X}$.
Recall the concentration property we assumed on $\logiczero(x)$ and $\logicone(x)$ and a good candidate of $\widehat{X}$ emerges: $\widehat{X} = -\frac{\dee^2}{\dee x^2} + \dwfunc$ where $\dwfunc$ is a double square well potential parameterized by $\ell, w, s > 0$ which is defined as follows.
$$
    \dwfunc(x) \coloneqq
    \left\{
        \begin{array}{ll}
            +\infty, & x \in (\ell + w/2,+\infty) \\
            0, & x \in [w/2, \ell + w/2] \\
            s, & x \in [0, w/2) \\
            \dwfunc(-x), & x < 0
        \end{array}
    \right.
$$
As we can see in \fig{1}, if $s$ is sufficiently large, the first two eigenfunctions of $\widehat{X} = -\frac{\dee^2}{\dee x^2} + \dwfunc$ is approximately $\logiczero(x) \pm \logicone(x)$ for $\logiczero(x)$ being a partial sine wave in $x>0$ and $\logicone(x) \coloneqq \logiczero(-x)$.

It turns out that a cleaner approach is to \emph{define} $\ket{\logiczero}$ and $\ket{\logicone}$ from $\widehat{X}$: Fix $\widehat{X} \coloneqq -\frac{\dee^2}{\dee x^2} + \dwfunc$ with the ground state $\ket{\psi_0}$ and the first excited state $\ket{\psi_1}$.
Let $\ket{\logiczero} \coloneqq \frac{1}{\sqrt{2}} \left( \ket{\psi_0}+\ket{\psi_1} \right)$ and $\ket{\logicone} \coloneqq \frac{1}{\sqrt{2}} \left( \ket{\psi_0}-\ket{\psi_1} \right)$.
We now have that $\widehat{X}$ acting on $\{ \ket{\logiczero},\ket{\logicone}\}$ exactly resembles $\frac{\delta}{2}X + c_2I$ acting on $\{\ket{0},\ket{1}\}$ for some $c_2$ where $\delta$ is the spectral gap of $\widehat{X}$.
The remaining problem is how well $\widehat{Z} \coloneqq \sgn(x)$ resembles $Z$ in the same manner.
It is easy to see from Figure~\ref{fig:1} that increasing $s$ will make $\ket{\logiczero}$ concentrate more in the region $x>0$ and hence the relation $\widehat{Z}\ket{\logiczero} \approx \ket{\logiczero}$ is approximated with smaller error.
However, this comes with a cost of reducing the spectral gap $\delta$.
Note that we want $\delta$ to be not too small because our reduction reads $X \mapsto \frac{2}{\delta} \widehat{X}$.
An exponentially small $\delta$ means the potential function hiding in the \sdg\ operator $\frac{2}{\delta} \widehat{X}$ must be exponentially large, which will make the whole reduction inefficient in the sense that the operator we reduced to involves prohibitively large energy.
Therefore, we need to pick an appropriate $s$ such that $\norm{\widehat{Z}\ket{\logiczero} - \ket{\logiczero}}$ is small enough but $\delta$ is not too small.
Studying the relation between $\delta$ and $s$ can be regarded as doing semiclassical analysis on the operator $-h^2\Delta + \dwfunc$ with $s=1$: this operator is clearly equivalent to $-\Delta + \dwfunc$ with $s = 1/h^2$.
This is why we will invoke results from semiclassical analysis in \sec{semiclassical} for the actual construction of $\widehat{X}$ in \sec{proof-of-lem-main}.

The idea above can be naturally generalized to higher dimensions.
Note that we used potential functions that are not continuous in the construction of the reduction $H \mapsto \widehat{H}$ above, and we need smooth potentials in the actual construction of the \sdg\ operator.

\subsection{1-Dimensional Double-Well Function} \label{sec:semiclassical}

The goal of this section is to establish a proof for a technical lemma (\lem{semiclassical-main}) which will be used later.
It characterizes the spectral property of a \sdg\ operator with the potential being a 1-dimensional double-well function.

\begin{definition}[Sobolev Space]
    Let $M$ be a finite-dimensional compact connected Riemannian manifold or an open subset of such a manifold.\footnote{Most $M$ encountered later in this section are simply $[-1,1]$ and its open subsets. The general definition is put here for the sake of completeness. }
    Define $H^2_0(M) \coloneqq W^{2,2}_0(M)$ as the Sobolev space regarding $\ell^2$-norm and $2$nd order derivatives, which consists of functions $f:M \to \mathbb{R}$ vanishing on the boundary $\partial M$.
\end{definition}

\paragraph{Operator of Interest.} Throughout this section, we focus on analyzing the property of \sdg\ operator $P \coloneqq -h^2 \Delta + V$ when $h \to 0$.
This is called \emph{semi-classical analysis} in literature.
Here, we are only interested in the case that $V$ is a fixed \emph{double-well function} defined as follows.

\begin{definition}[Double-Well Function] \label{def:dw}
    We call $V: [-1,1] \to \mathbb{R}$ a \emph{double-well function} if it satisfies the following:
    \begin{enumerate}
        \item $V$ is symmetric, i.e., $V(-x) = V(x)$ for any $x$.
        \item $V$ is smooth on $[-1,1]$.
        \item There exists some $a \in (0,1)$ such that $V$ has exactly two minima located at $a$ and $-a$ with $V(a) = V(-a) = 0$ and $V''(a) = V''(-a) > 0$.
    \end{enumerate}
\end{definition}

We restrict the domain of $P$ to be functions in $H^2_0([-1,1])$.
All kets introduced later in this section are \emph{real}-valued functions.
For why this works, note that e.g. $P$ is a real differential operator and hence all its eigenfunctions can be made real.

The so-called \emph{Agmon distance} is essential in semi-classical analysis:

\begin{definition}[Agmon Distance]
    Let $M$ be a finite-dimensional compact connected Riemannian manifold.
    Let $V: M \to \mathbb{R}$ be a potential function and $E \in \mathbb{R}$ be some specific energy of interest.
    The \emph{Agmon distance} between $x,y \in M$ is
    $$
        \agmonfull{V}{E}{x}{y} \coloneqq \inf_{\gamma} \int_{\gamma} \sqrt{\max\{V(z)-E,0\}} \dee z,
    $$
    where $\gamma$ runs over all piecewise $C^1$ paths connecting $x$ and $y$.
    The Agmon distance between a point $x \in M$ and a subset $Y \in M$ follows naturally:
    $$
        \agmonfull{V}{E}{x}{U} \coloneqq \inf_{y \in Y} \agmonfull{V}{E}{x}{y}.
    $$
    We define $S_0 \coloneqq \agmonfull{V}{0}{a}{-a}$ as it will make appearances throughout.
    Here $V$ is the fixed double-well function and $a$ and $-a$ are two minima of $V$ (Refer to Definition~\ref{def:dw}).
\end{definition}

The notion of \emph{Dirichlet realization} will be helpful when we study local behavior of a \sdg\ operator:

\begin{definition}[Dirichlet Realization]
    Let $P\coloneqq-h^2\Delta + V$ be a \sdg\ operator with $V: M \to \mathbb{R}$, where $M$ is a finite-dimensional compact connected Riemannian manifold.
    Let $\Omega \subseteq M$ be a bounded open set.
    $P_{\Omega}$ is defined as the restriction of $P$ over $\Omega$.
    To be more specific, $P$ has domain $H_0^2(\Omega)$.
\end{definition}

\begin{definition}[Ball]
    Let $M$ be a Riemannian manifold.
    Let $d(\cdot,\cdot)$ be a distance function on $M$.
    Define $\ball{d}{x}{r} \coloneqq \{ z \in M \ | \ d(z, x) \leq r \}$ and $\ballopen{d}{x}{r} \coloneqq \{ z \in M \ | \ d(z, x) < r \}$.
\end{definition}

Recall that our target is to analyze the behavior of $P$ when $h \to 0$.
As such, the $k$th eigenvalue of $P$ tends to $0$ for any fixed $k$.
Therefore we need the following notion to keep track of those eigenvalues we are interested in when $h \to 0$.
\begin{definition}[Eigenvalue Interval, $I(h)$]
    $I(h)$ is defined as a closed interval in $\mathbb{R}$ that depends on the value of $h$.
    We require that $I(h) \to \{0\}$ when $h \to 0$.
\end{definition}

\begin{definition}[Local Eigenvalues and Eigenstates]
    Fix a small $\eta > 0$.
    Define intervals $L, R \subseteq [-1,1]$ as follows:
    \begin{align*}
        L &\coloneqq \ballopen{\agmonsymbol}{-a}{3S_0} - \ball{\agmonsymbol}{a}{\eta} - [a,+\infty] \\
        R &\coloneqq \ballopen{\agmonsymbol}{a}{3S_0} - \ball{\agmonsymbol}{-a}{\eta} - [-\infty,-a] 
    \end{align*}
    where $\agmonsymbol \coloneqq \agmonsymbolandparams{V}{0}$.
    Let $J=L$ or $J=R$. 
    Denote by $\mu_{J,1},\mu_{J,2},\dots, \mu_{J, m_J}$ eigenvalues\footnote{Multiplicity is included here, i.e., degenerate eigenvalues appear multiple times in the list.} of $P_J$ in interval $I(h)$ and by $\ket{\phi_{J,1}},\ket{\phi_{J,2}},\dots, \ket{\phi_{J,m_J}}$ their corresponding normalized eigenstates that constitute an orthogonal system.
    Note that $P_L$ is just the reflection of $P_R$ across the y-axis.
\end{definition}

\begin{definition}[Hidden constants in big O notation]
    The subscript(s) $c$ in the notation $\bigOdep{c}{\cdot}$ means that the hidden constant depends on $c$.
\end{definition}

\begin{lemma} \label{lem:semiclassical-main}
    Let $P \coloneqq -h^2 \Delta + V$ be a \sdg\ operator with some double-well function $V$.
    Let $\mathcal{F}$ be the subspace spanned by eigenstates of $P$ whose eigenvalues are in $I(h)$. The following holds for sufficiently small $h > 0$ if there exists $r(h) \geq \Omega(e^{-o(1/h)})$ such that $P_L$ and $P_R$ have no eigenvalues $r(h)$-close to $I(h)$ from outside, i.e., no eigenvalues in $(I(h) + \ball{\agmonsymbol}{0}{r(h)}) - I(h)$:
    \begin{enumerate}[label=(\roman*)]
        \item There exists a bijection $\mathfrak{b}$ mapping eigenvalues of $P$ in $I(h)$ to eigenvalues of $P_L$ and $P_R$ in $I(h)$. Furthermore, $\abs{\mathfrak{b}(\lambda) - \lambda} \leq \bigOdep{\eta}{e^{-\frac{S_0 -2\eta - o(1)}{h}}}$.
        \item For each eigenstate $\ket{\phi_{J,k}}$ $(J\in\{L,R\})$, there exists a $\ket{v_{J,k}} \in \mathcal{F}$ such that $\norm{\ket{v_{J,k}} - \ket{\phi_{J,k}} }\leq \bigOdep{\eta}{e^{-\frac{S_0 - \epsilon(\eta)}{h}}}$ where $\epsilon(\eta)$ is some function such that $\lim_{\eta \to 0}\epsilon(\eta) = 0$. 
        $v_{L,k}(\cdot)$ is $v_{R,k}(\cdot)$ reflected across the y-axis.
        Furthermore, $\{\ket{v_{J,k}}\}$ are almost orthogonal to each other in the sense that $\abs{\braket{v_{J,k}|v_{J',k'}}} \leq \bigOdep{\eta}{e^{-\frac{S_0 - \epsilon(\eta)}{h}}}$ for any $(J,k) \neq (J',k')$.
    \end{enumerate}
\end{lemma}
\begin{remark}
    \lem{semiclassical-main} is a special case of Proposition A.3 and the following discussion in Liu et al.~\cite{liu2023quantum}.
    Refer to \cite[Section 6.c]{dimassi1999spectral} and \cite[Section 4.2]{helffer2006semi} for more complete treatments.
    To be more specific, first, it is easy to verify that $V$ satisfies Assumption A.3 and A.4 in~\cite{liu2023quantum}.
    Therefore the first part of \lem{semiclassical-main} directly follows from Proposition A.3.
    Note that our assumption in \lem{semiclassical-main} does not ask for spectral property of $P$ like Proposition A.3 does in~\cite{liu2023quantum}: the assumption is unnecessary as pointed out in Remark 4.2.3 of~\cite{helffer2006semi}.
    Readers may have noticed that the definition of $R$ (similar for $L$) in \cite{liu2023quantum} should be $R \coloneqq \ballopen{\agmonsymbol}{a}{3S_0} - \ball{\agmonsymbol}{-a}{\eta}$, which is a disjoint union of two intervals (for small enough $\eta$) while ours discards the interval on which $V(x)$ has a nonzero lower bound that depends on $\eta$.
    This difference can be justified by that we are interested in the case when $h \to 0$ throughout:
    $P$ restricted on the other interval cannot have nontrivial eigenfunction with eigenvalue in $I(h)$ since $I(h) \to \{0\}$ and $V(x)$ on that interval is bounded below by a positive value.
    For the second part of \lem{semiclassical-main}, the orthogonality comes from Eq. (189) and (190) of~\cite{liu2023quantum}, while the bound of $\norm{\ket{v_{J,k}} - \ket{\phi_{J,k}} }$ can be derived from Eq. (187) of~\cite{liu2023quantum} and the decay of $\phi_\alpha, \psi_\alpha$ given in p. 59 of~\cite{dimassi1999spectral}. 
\end{remark}

\begin{lemma}[{\cite[Proposition 6.4]{dimassi1999spectral}}] \label{lem:single-well-concentration}
    Let $M$ be a finite-dimensional compact connected Riemannian manifold.
    Fix an smooth $V:M \to \mathbb{R}$.
    Define the energy well $U \coloneqq \{\vect{x} \ | \ V(\vect{x}) \leq 0 \}$.
    Let $h \to 0$ and assume that $u = u(\vect{x};h)$ is an eigenfunction of $P \coloneqq -h^2 \Delta + V$ with eigenvalue $o(1)$.
    Then,
    $$
        \norm{e^{\frac{\agmonfull{V}{E}{\vect{x}}{U}}{h}}u} \coloneqq \sqrt{\int_{M} \abs{ e^{\frac{\agmonfull{V}{E}{\vect{x}}{U}}{h}}u (\vect{x}) }^2 \dee \vect{x}} = \bigO{e^{\frac{o(1)}{h}}}.\footnote{The orginal theorem statement in~\cite{dimassi1999spectral} is that for every $\epsilon$, one have for small enough $h>0$ depending on $\epsilon$ such that the left-hand side is no more than $C_{\epsilon} e^{\epsilon/h}$ where $C_{\epsilon}$ is a constant that depends on $\epsilon$. This clearly induces the statement in \lem{single-well-concentration} by letting $\epsilon \to 0$.}
    $$
\end{lemma}

\begin{lemma}[{\cite[Theorem 11.3]{hislop2012introduction}}] \label{lem:harmonic}
    Fix $H(\lambda) \coloneqq - \Delta + \lambda^2 V$ with $V: \mathbb{R}^n \to \mathbb{R}$ that satisfies the following:
    \begin{enumerate}[label=(\roman*)]
        \item $V \in C^3(\mathbb{R}^n)$; $V \geq 0$; $\liminf_{\abs{\vect{x}} \to \infty} V(\vect{x}) = +\infty$.
        \item $V$ has a single, nondegenerate zero at $\vect{x} = \vect{0}$, i.e., the Hessian of $V$ at $\vect{0}$, denoted by $A$, is positive semidefinite.
    \end{enumerate}
    Let $K \coloneqq -\Delta + \frac12 \vect{x}^\intercal A \vect{x}$.
    We list eigenvalues of $K$, including multiplicity, as $e_0 \leq e_1 \leq e_2 \leq \cdots$.
    Similarly, label eigenvalues of $H(\lambda)$ as $e_0(\lambda) \leq e_1(\lambda) \leq e_2(\lambda) \leq \cdots$.
    Then, for any $k$ we have
    $$
        e_k(\lambda) = \lambda e_k + \bigOdep{k}{\lambda^{\frac{4}{5}}}, \  \text{as } \lambda \to +\infty.
    $$
\end{lemma}
\begin{corollary} \label{cor:harmonic}
    Under the same assumption in \lem{harmonic} with $H(h) \coloneqq -h^2\Delta + V$, we have
    $$
        e_k(h) = h e_k + \bigOdep{k}{h^{\frac{6}{5}}}, \  \text{as } h \to 0.
    $$
\end{corollary}

\begin{lemma}[\cite{nakamura1986remark}] \label{lem:1d-double-well-gap}
    Fix $H(\lambda)\coloneqq-\Delta + \lambda^2 V$ with $V:\mathbb{R} \to \mathbb{R}$ that satisfies the following:
    \begin{enumerate}[label=(\roman*)]
        \item $V$ is a nonnegative continuous function such that $V(x)$ takes its absolute minimum zero at some $x=a$ and $x=b$ for some $a,b \in \mathbb{R}$.
        \item $\liminf_{\abs{\vect{x}} \to \infty} V(x) > 0$.
    \end{enumerate}
    Denote by $e_{n-1}(k)$ the $n$th smallest eigenvalue of $H(\lambda)$.
    Fix an arbitrary $n$, we have
    $$
        \liminf_{k \to \infty} \frac{1}{k} \ln (e_{n}(k) - e_{n-1}(k)) \geq -\int_{a}^b \sqrt{V(x)} \dee x.
    $$
\end{lemma}

\begin{lemma}[{\cite[Theorem 12.3]{hislop2012introduction}}]  \label{lem:1d-double-well-gap2}
    Fix $H(\lambda)\coloneqq-\Delta + \lambda^2 V$ with $V:\mathbb{R} \to \mathbb{R}$ that satisfies the following:
    \begin{enumerate}[label=(\roman*)]
        \item $V$ is symmetric, i.e., $V(x) = V(-x)$.\footnote{This induces Assumption (A3) in the description of Theorem 12.3 in~\cite{hislop2012introduction}.}
        \item  $V \in C(\mathbb{R}^n)$; $V \geq 0$; $\liminf_{\abs{x} \to \infty} V(x) = +\infty$.
        \item $V(x)$ has two nondegenerate minima at some $x=-a$ and $x=a$ for some $a \in \mathbb{R}$.
    \end{enumerate}
    Denote by $e_{n-1}(k)$ the $n$th smallest eigenvalue of $H(\lambda)$.
    Fix an arbitrary $n$, we have
    $$
        \limsup_{k \to \infty} \frac{1}{k} \ln (e_{1}(k) - e_{0}(k)) \leq -\int_{a}^b \sqrt{V(x)} \dee x.
    $$
\end{lemma}

Combining \lem{1d-double-well-gap} and \ref{lem:1d-double-well-gap2} we immediately have the following:

\begin{corollary} \label{cor:1d-double-well-gap}
    Under the same assumption in \lem{1d-double-well-gap2} with $H(h) \coloneqq -h^2\Delta + V$, we have
    $$
        e_1(k) - e_{0}(k) \subseteq \bigOmega{e^{-\frac{S_0 + o(1)}{h}}} \cap \bigO{e^{-\frac{S_0 - o(1)}{h}}}, \quad S_0 \coloneqq \agmonfull{V}{0}{a}{b}.
    $$
\end{corollary}

\begin{remarknumbered} \label{rem:extend}
    Corollary~\ref{cor:harmonic} (resp., Corollary~\ref{cor:1d-double-well-gap}) can be extended to the case where $V$ is defined only on a compact or a bounded open subset $M \subseteq \mathbb{R}^n$ (resp., $M \subseteq \mathbb{R}$) which is connected: in such a case we require eigenfunctions to vanish on the boundary of $M$ and we no longer need any assumption for $\liminf_{\abs{x} \to \infty} V(x)$. To understand this, it is important to note that the spectrum will be unchanged if we extend $V$ to $\widehat{V}$, where
    $$
        \widehat{V}(\vect{x}) = \left\{
            \begin{array}{ll}
                V(\vect{x}), & x \in M; \\
                +\infty, & x \notin M.
            \end{array}
        \right.
    $$
    Now that $\widehat{V}$ is defined on $\mathbb{R}^n$, the only obstacle for applying Corollary~\ref{cor:harmonic} is that $\widehat{V}$ is not continuous, let alone smooth.
    This can be tackled by a standard method that constructs a smooth function that is sufficiently close to $\widehat{V}$.
\end{remarknumbered}

\begin{theorem} \label{thm:double-well-main}
    Let $E_0 \leq E_1 \leq E_2 \leq \dots$ be eigenvalues of $P \coloneqq -h^2 \Delta + V$ with a double-well function $V$ and $\{\ket{\psi_k}\}_{k \in \mathbb{N}}$ their corresponding eigenstates.
    \begin{enumerate}[label=(\roman*)]
        \item $E_1 - E_0 \subseteq \bigOmega{e^{-\frac{S_0 + o(1)}{h}}} \cap \bigO{e^{-\frac{S_0 - o(1)}{h}}}$.
        \item $E_2 - E_1 \geq \Omega(h)$.
        \item There exists a $c \in  \{\pm1\}$ such that $\ket{\logiczero} \coloneqq \frac{1}{\sqrt2} \left( \ket{\psi_0} + c \ket{\psi_1} \right)$ is bounded away from $x < a$. To be more specific, the following inequality holds:
        $$
            \forall x_0 \in (-a,a),\quad \int_{-1}^{x_0} \abs{\logiczero(x)}^2 \dee x \leq \bigO{e^{-(2-o(1))\frac{\agmon{a}{x_0}}{h}}},
        $$
        where $\agmonsymbol\coloneqq\agmonsymbolandparams{V}{0}$.
        A similar inequality holds for $\ket{\logicone} \coloneqq \frac{1}{\sqrt2} \left( \ket{\psi_0} - c \ket{\psi_1} \right)$ since it is just $\ket{\logiczero}$ reflected across the y-axis.
    \end{enumerate}
\end{theorem}
\begin{proof}[Proof of \thm{double-well-main}]
    We prove (i), (ii) and (iii) in order.
    \paragraph{(i).}
    This comes from directly applying Corollary~\ref{cor:1d-double-well-gap} with Remark~\ref{rem:extend}.
    
    \paragraph{(ii).}
    Recall that $P_L$ is just the reflection of $P_R$ across the y-axis. 
    By Corollary~\ref{cor:harmonic} and Remark~\ref{rem:extend}, the eigenvalues of $P_L$ and $P_R$ are 
    $$
        e_k(h) = (2k+1)qh+\bigOdep{k}{h^{6/5}},\quad k =0,1,2,\dots,
    $$ where $q \coloneqq V''(-a) = V''(a)$.
    Therefore, by setting interval $I(h)$ to be $[\frac12 e_0(h), \frac12(e_1(h)+e_2(h))]$ it satisfies the condition in \lem{semiclassical-main}.
    \lem{semiclassical-main}(i) tells us $E_{2k+1/2\pm1/2}= e_k(h) \pm \bigOdep{\eta}{e^{-\frac{S_0 -2\eta - o(1)}{h}}}$. 
    Therefore, (ii) is obtained by picking a sufficiently small $\eta$.

    \paragraph{(iii).}
    Similar to what we have done for (ii), applying \lem{semiclassical-main} with $I(h) = [\frac12 e_0(h), \frac12(e_0(h)+e_1(h))]$ we obtain a symmetric pair $\ket{v_{J}} \coloneqq \ket{v_{J,0}}$ for $J \in \{L,R\}$ from symmetric pair $\ket{\phi_{J}} \coloneqq \ket{\phi_{J,0}}$.
    Note that $\ket{v_J} \in \mathcal{F} = \mathrm{span}(\ket{\psi_0}, \ket{\psi_1})$.
    By varying $\eta$ from \lem{semiclassical-main}(ii) we have the following holds for any $\varepsilon > 0$:
    \begin{equation} \label{eq:1d-double-well.ortho}
        \begin{gathered}
        \norm{v_J} \in \norm{\phi_J} \pm \norm{v_J - \phi_J} \in 1 \pm \bigOdep{\varepsilon}{e^{-\frac{S_0-\varepsilon}{h}}}, \\
        \abs{\braket{v_L | v_R}} \leq \bigOdep{\varepsilon}{e^{-\frac{S_0-\varepsilon}{h}}}.
        \end{gathered}
    \end{equation}
    For convenience we use notation $\bigO{r}$ to represent $\bigOdep{\varepsilon}{e^{-\frac{S_0-\varepsilon}{h}}}$ below.
    Now, due to the symmetry between $\ket{v_L}$ and $\ket{v_R}$ we can assume for some $\alpha,\beta \in \mathbb{R}$ that
    \begin{equation} \label{eq:1d-double-well.factor}
        \begin{gathered}
            \ket{v_L} = \alpha \ket{\psi_0} - \beta \ket{\psi_1}, \\
            \ket{v_R} = \alpha \ket{\psi_0} + \beta \ket{\psi_1},
        \end{gathered}
    \end{equation}
    because $\psi_0(\cdot)$ is symmetric and $\psi_1(\cdot)$ antisymmetric regarding $x = 0$.
    Plugging \eqn{1d-double-well.factor} into \eqn{1d-double-well.ortho} we get $\alpha^2 + \beta^2 \in 1 \pm \bigO{r}$ and $\alpha^2 - \beta^2 \in \pm \bigO{r}$.
    Therefore,
    $$
        \abs{\alpha},\abs{\beta} \in \frac{1}{\sqrt2} \pm \bigO{r}.
    $$
    Without loss of generality let $\alpha  \in \frac{1}{\sqrt2} \pm \bigO{r}$.
    Therefore $\ket{v_R} = \alpha \ket{\psi_0} + c \abs{\beta} \ket{\psi_1}$ for some $c \in \{ \pm1\}$, which implies
    $
        \ket{\logiczero} - \ket{v_R} \in \pm \bigO{r} \ket{\psi_0} \pm \bigO{r} \ket{\psi_1}.
    $
    
    \begin{claim} \label{clm:zero-to-v}
        $\norm{\ket{\logiczero} - \ket{v_R}} = \bigO{r}$.
    \end{claim}

    The next claim is immediate from \lem{semiclassical-main}(ii).

    \begin{claim} \label{clm:v-to-phi}
        $\norm{\ket{v_R} - \ket{\phi_R}}= \bigO{r}$. 
    \end{claim}

    Now that $\ket{\phi_R}$ is clearly a good estimate of $\ket{\logiczero}$.
    Our plan is to show the inequality in (iii) with $\ket{\logiczero}$ replaced by $\ket{\phi_R}$.
    Recall that $\ket{\phi_R}$ is the ground state of $P_R$.
    Note that we are able to apply \lem{single-well-concentration} on $P_R$ despite $R$ not being a finite-dimensional compact Riemannian manifold: 
    It does the work by naturally extending the domain of $V$ in \lem{single-well-concentration} from $R$ to $\overline{R}$ and nothing will change effectively.
    Set $U = \{a\}$ in \lem{single-well-concentration}.
    The eigenvalue for $\ket{\phi_R}$ is $\Theta(h) \leq o(1)$ as discussed in the proof of (ii).
    Set $\agmonsymbol \coloneqq \agmonsymbolandparams{V}{0}$ as usual.
    Extend the definition of $\phi_R(x)$ by letting $\phi_R(x) = 0$ for $x \notin R$.
    We have
    \begin{align*}
        \int_{-1}^{x_0} \abs{\phi_R(x)}^2 \dee x &\leq e^{-2\frac{\agmon{a}{x_0}}{h}} \int_{-1}^{x_0} \abs{ e^{\frac{\agmon{a}{x}}{h}}\phi_R(x)}^2 \dee x & (\agmon{a}{x} \geq \agmon{a}{x_0})\\
        &\leq e^{-2\frac{\agmon{a}{x_0}}{h}} \int_R \abs{ e^{\frac{\agmon{a}{x}}{h}}\phi_R(x)}^2 \dee x \\
        &\leq e^{-2\frac{\agmon{a}{x_0}}{h}} \bigO{e^{\frac{o(1)}{h}}} & (\textrm{\lem{single-well-concentration}}) \\
        &\leq \bigO{e^{-2(1-o(1))\frac{\agmon{a}{x_0}}{h}}}.
    \end{align*}
    Combining the above with \clm{zero-to-v} and \ref{clm:v-to-phi} we obtain
    $$
        \sqrt{\int_{-1}^{x_0} \abs{\logiczero(x)}^2 \dee x} \leq  \sqrt{\int_{-1}^{x_0} \abs{\phi_R(x)}^2}  + \bigO{r} \leq \bigO{e^{-(1-o(1))\frac{\agmon{a}{x_0}}{h}}} + \bigOdep{\varepsilon}{e^{-\frac{S_0-\varepsilon}{h}}}
    $$
    from triangle inequality
    $$
         \sqrt{\int_{-1}^{x_0} \abs{\logiczero(x)}^2 \dee x} -  \sqrt{\int_{-1}^{x_0} \abs{\phi_R(x)}^2} \leq \sqrt{\int_{-1}^{x_0} \abs{\logiczero(x) - \phi_R(x)}^2 \dee x} \leq \norm{\logiczero(x) - \phi_R(x)} \leq \bigO{r}.
    $$
    Note that the latter term $\bigOdep{\varepsilon}{e^{-\frac{S_0-\varepsilon}{h}}}$ can be absorbed into the former term by choosing sufficiently small $\varepsilon$, since $\agmon{a}{x_0} < \agmon{a}{-a} = S_0$ from i) $x_0 \in (-a,a)$ and that ii) $\widehat{V}(x)>0$ for any $x \in \mathbb{R}-\{a\}$. 
\end{proof}

\subsection{Perturbation Theory} \label{sec:pert}

We first introduce an alternative way to identify eigenvalues for operators that may not have purely discrete spectrum.

\begin{definition} \label{def:eigen}
    Let $H$ be a self-adjoint operator that is bounded from below, i.e., $H \succeq cI$ for some $c$.
    For $k \in \mathbb{N}$, define
    $$
        \eigen_k(H) \coloneqq \sup_{\ket{\phi_0}, \dots, \ket{\phi_{k-1}} \in \dom(H)} \inf_{\substack{\ket{\psi} \in \dom(H) \\ {\norm{\ket{\psi}}=1} \\\ket{\psi} \perp \ket{\phi_0},\dots,\ket{\phi_{k-1}} }} {\braket{\psi | H | \psi}}.
    $$
    For notational convenience, We also define $\mineigen(H) \coloneqq \eigen_0(H)$ and $\maxeigen(H) \coloneqq \eigen_{m-1}(H)$ where $m \coloneqq \dim(H)$.
\end{definition}

The $\eigen_k(H)$ above are basically eigenvalues due to the following lemma:

\begin{lemma}[Min-Max Principle~{\cite[Theorem XIII.1]{reed1972methods}}]
    $\eigen_k(H)$ is the $(k+1)$th smallest eigenvalue counting multiplicity if it is below the bottom of the essential spectrum of $H$.
    Otherwise, $\eigen_k(H)$ is the bottom of the essential spectrum.
\end{lemma}

\begin{lemma}[Weyl's Inequality] \label{lem:weyl}
    Let $A$ and $B$ be self-adjoint, semi-bounded operators of the same domain.
    Then, for any $k$ we have
    $$
        \abs{\eigen_k(A) - \eigen_k(B)} \leq \norm{A-B}.
    $$
\end{lemma}
\begin{proof}
    By Definition~\ref{def:eigen},
    \begin{align*}
        \eigen_k(A) & = \sup_{\ket{\phi_0}, \dots, \ket{\phi_{k-1}} \in \dom(A)} \inf_{\substack{\ket{\psi} \in \dom(A) \\ {\norm{\ket{\psi}}=1} \\\ket{\psi} \perp \ket{\phi_0},\dots,\ket{\phi_{k-1}} }} \braket{\psi | A | \psi}  \\
        &= \sup \inf \left( \braket{\psi | B | \psi} + \braket{\psi | (A - B) | \psi} \right) \\
        & \geq \sup \left( \inf  \braket{\psi | B | \psi} + \inf \braket{\psi | (A - B) | \psi} \right) \\
        & \geq \sup \left( \inf  \braket{\psi | B | \psi} - \norm{A-B} \right) \\
        & = \eigen_k(B) - \norm{A-B}.
    \end{align*}
    Switching the role of $A$ and $B$ in the above we get $\eigen_k(B) \geq \eigen_k(A) - \norm{B-A}$, which concludes the proof.
\end{proof}

\begin{lemma} \label{lem:pert-spec}
    Let $H$ be a (possibly infinite-dimensional) Hamiltonian written as $H = \bigl( \begin{smallmatrix}A & R^\dagger\\ R & B\end{smallmatrix}\bigr)$ w.r.t. subspaces $\mathcal{S}$ and $\mathcal{S}^\perp$ such that $m\coloneqq\dim(\mathcal{S}) < \infty$.
    Suppose $\Delta \coloneqq \mineigen(B) - \maxeigen(A) > 0$.
    Then, for any $k \in \{0,1,\dots,m-1\}$,
    $$
        \abs{\eigen_k(H) - \eigen_k(A)} \leq \norm{R}. 
    $$
\end{lemma}
\begin{proof}
    Define $H_0 \coloneqq H - \bigl( \begin{smallmatrix}0 & R^\dagger\\ R & 0\end{smallmatrix}\bigr)$.
    Applying \lem{weyl} we have
    $$
        \abs{\eigen_k(H) - \eigen_k(H_0)} \leq \norm{\bigl( \begin{smallmatrix}0 & R^\dagger\\ R & 0\end{smallmatrix}\bigr)} = \norm{R}
    $$
    for any $k \in \mathbb{N}$.
    It remains to show that $\eigen_k(H_0) = \eigen_k(A)$ for $k \in \{0,1,\dots,m-1\}$: 
    Note that the spectrum of $H_0 = \bigl( \begin{smallmatrix}A & 0\\ 0 & B\end{smallmatrix}\bigr)$ is the union of those of $A$ and $B$.
    And the assumption $\Delta > 0$ guarantees that the whole spectrum of $B$ is above that of $A$.
    Hence  $\eigen_k(H_0) = \eigen_k(A)$ for $k \in \{0,1,\dots,m-1\}$.
\end{proof}

\begin{lemma} \label{lem:pert-sim-diff}
    Let $A$ and $B$ be two self-adjoint operators.
    Then,
    $$
        \norm{e^{-iAt} - e^{-iBt}} \leq \norm{A-B}t.
    $$
\end{lemma}
\begin{proof}
    By Duhamel's Principle or differentiating $e^{-i(A-B)t}$ we have
    $$
        e^{-iAt} - e^{-iBt} = -i \int_0^t e^{-iA(t-s)} (A - B) e^{-iBs}  \dee s.
    $$
    The inequality follows by taking norms of both sides.
\end{proof}

\begin{lemma} \label{lem:pert-sim}
    Let $H$ be a (possibly infinite-dimensional) Hamiltonian written as $H = \bigl( \begin{smallmatrix}A & R^\dagger\\ R & B\end{smallmatrix}\bigr)$ w.r.t. subspaces $\mathcal{S}$ and $\mathcal{S}^\perp$.
    For any $t > 0$ we have 
    $$
        \norm{P_{\mathcal{S}}e^{-iHt}P_{\mathcal{S}} - e^{-iAt}} \leq \frac{2\sqrt{2}}{3}\left(\norm{R}t\right)^{3/2}.
    $$
\end{lemma}
\begin{proof}
    The proof largely follows that of Theorem 10 in~\cite{leng2024expanding} with the exception that we cannot use the vanilla Schrieffer--Wolff theory to prove \clm{trunc-leak} because our case involves unbounded (and hence infinite-dimensional) operators.
    
    Let $H_0 = P_\S H P_\S = \diag(A,\vect{0})$, and we define:
\begin{align*}
    \mathscr{A}(t) = P_\S e^{-iHt}P_\S,\quad \mathscr{B}(t) = P_\S e^{-iH_0 t}P_\S.
\end{align*}
We also define the $\mathscr{E}(t) \coloneqq \mathscr{A}(t) - \mathscr{B}(t)$. Note that $\mathscr{E}(0) = 0$. We find that
\begin{align*}
    \frac{\d}{\d t}\mathscr{A}(t) &= P_\S \left(-i H\right)\left(P_\S + P_{\S^\perp}\right) e^{-iHt} P_\S= -i H_0\mathscr{A}(t) -i R e^{-iHt} P_\S,\\
    \frac{\d}{\d t}\mathscr{B}(t) &= P_\S \left(-iH_0\right)e^{-iH_0 t}P_\S=-i H_0 \mathscr{B}(t).
\end{align*}
Therefore, we have $\dot{\mathscr{E}}(t) =\dot{\mathscr{A}}(t)-\dot{\mathscr{B}}(t) = -i H_0 \mathscr{E}(t) - i Re^{-iHt} P_\S$. By Duhamel's principle we represent the function $\mathscr{E}(t)$ as the time integral:
\begin{align*}
    \mathscr{E}(t) = -i\int^t_0 \exp(-iH_0 (t-\tau))R \exp(-iH\tau) P_{\S} ~\d\tau,
\end{align*}
so we can estimate the error (recall that $R = P_\S H P_{\S^\perp}$):
\begin{align*}
    \|\mathscr{E}(t)\| &\le \int^t_0 \|P_\S H P_{\S^\perp} \exp(-iH\tau) P_{\S}\| ~\d \tau \le \int^t_0 \|RP_{\S^\perp} \exp(-iH\tau) P_{\S}\|~\d \tau \\
    &\le \int^t_0 \|R\|\left(\sqrt{2 \norm{R} \tau}\right)~\d \tau = \frac{2\sqrt{2}}{3}\left(\norm{R}t\right)^{3/2}.
\end{align*}
In the second step, we use $P_{\S^\perp} = P^2_{\S^\perp}$. In the second to last step, we invoke \clm{trunc-leak}.
    \begin{claim} \label{clm:trunc-leak}
        $\norm{P_{\mathcal{S}^\perp} \exp(-iHt) P_{\mathcal{S}}} \leq \sqrt{2 \norm{R} t}$.
    \end{claim}
    \clm{trunc-leak} is equivalent to a $2 \norm{R} t$ upper bound of $\norm{P_{\mathcal{S}^\perp} \exp(-iHt) \ket{\psi_0}}^2$ for every normalized state $\ket{\psi_0} \in \mathcal{S}$.
    We prove such a bound below.
    Define $\ket{\psi} \coloneqq \ket{\psi(t)} \coloneqq \exp(-iH t) \ket{\psi_0}$.
    We have
    $$
        \frac{\dee}{\dee t} P_{\mathcal{S}}\ket{\psi} = -i P_{\mathcal{S}} H \ket{\psi}
    $$
    since $\frac{\dee}{\dee t} \ket{\psi} = -i H \ket{\psi}$.
    Therefore,
    \begin{align*}
        \frac{\dee}{\dee t} \norm{P_{\mathcal{S}} \ket{\psi}}^2 &= 2 \Re{(-iP_{\mathcal{S}} H \ket{\psi})^\dagger (P_{\mathcal{S}} \ket{\psi})} \\
        &= - 2 \Im{\braket{\psi | H P_{\mathcal{S}} | \psi }} \\
        &= - 2 \Im{ \braket{\psi | P_{\mathcal{S}} H P_{\mathcal{S}} | \psi } + \braket{\psi | P_{\mathcal{S}^\perp} H P_{\mathcal{S}} | \psi } } \\
        & = - 2 \Im{  \braket{\psi | R^\dagger | \psi } }.
    \end{align*}
    Consequently,
    \begin{align*}
        \norm{P_{\mathcal{S}^\perp} \exp(-iHt) \ket{\psi_0}}^2 &= \norm{P_{\mathcal{S}^\perp} \ket{\psi(t)}}^2  \\
        &= 1 - \norm{P_{\mathcal{S}} \ket{\psi(t)}}^2 \\
        &=  \norm{P_{\mathcal{S}} \ket{\psi(0)}}^2 - \norm{P_{\mathcal{S}} \ket{\psi(t)}}^2 \\
        &= \int_{t}^0 \left( \frac{\dee}{\dee \tau} \norm{P_{\mathcal{S}} \ket{\psi(\tau)}}^2 \right)\dee \tau \\
        &= \int_{t}^0 -2 \Im{  \braket{\psi(\tau) | R^\dagger | \psi(\tau) } } \dee \tau \\
        & \leq 2 \norm{R} t.
    \end{align*}
\end{proof}

\subsection{Proof of the Main Lemma} \label{sec:proof-of-lem-main}
This section is dedicated to the proof of \lem{main}.

\begingroup
    \def\thetheorem{\ref{lem:main}}
    \begin{lemma}[restated]
    
    Fix $M_1,M_2,\epsilon_1 > 0$ and $t \geq 0$.
    For any TIM Hamiltonian 
    $$
        H = \sum_{1 \leq u \leq n} (a_u X_u + b_u Z_u) + \sum_{1 \leq u,v \leq n} b_{u,v} Z_u Z_v
    $$
    such that $|a_u|,|b_u|,|b_{u,v}| \leq M_1$ and $a_u > 1/M_2$, 
    there exists a \sdg\ operator
    $$
        \widehat{H} = - G^\star \Delta + V
    $$
    and a polynomial $P \leq \poly(n,M_1,M_2,1/\epsilon_1,t)$ satisfying the following:
    \begin{enumerate}[label=(\roman*)]
        \item $1 \leq G^\star \leq P$ is $\poly(n)$-time computable and the potential $V: [-1,1]^n \to \mathbb{R}$ is $2$-body, smooth, $\poly(n)$-time computable and has $C^1$-norm bounded within $P$.
        \item $\abs{\eigen_k(\widehat{H}) - \eigen_k(H)} \leq \epsilon_1$ for every $k \in \{0,1,\dots,2^n-1\}$.
        \item $\norm{ W^\dagger P_{\mathcal{S}} e^{-i\widehat{H}t}W - e^{-iHt}} \leq \epsilon_1$, where $W \coloneqq \bigotimes_{1 \leq k \leq n} (\ket{\psi_0^{k}}\bra{-} + \ket{\psi_1^{k}}\bra{+})$ shifts basis and $\mathcal{S}$ is the image of $W$.  $\ket{\psi_0^k}$ and $\ket{\psi_1^k}$ are $1$-dimensional.
        Moreover, wave functions $\psi_j^k(\cdot)$ are smooth, $\poly(n)$-time computable and have $C^1$-norm bounded within $\polylog(P)$.
    \end{enumerate}

    \end{lemma}
    \addtocounter{theorem}{-1}
\endgroup

\paragraph{1D Building Blocks.}
Fix $\dwfunc(x) \coloneqq (x-1/2)^2(x+1/2)^2$ to be a quartic function with minima $x=\pm x^\star\coloneqq\pm1/2$.\footnote{Other choices of $\dwfunc(x)$ should work in principle but it is trickier to assert the \emph{quantitative} $C^1$-boundedness of its eigenfunctions, which will be mentioned soon after.}
It is easy to verify that $\dwfunc$ is a double-well function that satisfies the requirement of $V$ in \thm{double-well-main}.
Define the operator $\widehat{X}(h)$, which depends on a parameter $h > 0$, to be
$$
    \widehat{X}(h) \coloneqq C_h \left( -h^2 \frac{\dee^2}{{\dee x}^2} + \dwfunc \right),
$$
where $C_h$ is chosen such that the spectral gap of $\widehat{X}(h)$ is $2$.
Let $G_\circ \coloneqq e^{\frac{S_0}{h}}$ where $S_0 \coloneqq \agmonfull{\dwfunc}{0}{x^\star}{-x^\star}$.
Let $E_0 \leq E_1 \leq E_2 \leq \cdots$ be eigenvalues of $\widehat{X}(h)$ and $\ket{\psi_0},\ket{\psi_1},\ket{\psi_2},\dots$ be corresponding eigenstates given $h$.
It is well known that eigenfunctions of an elliptic operator, of which \sdg\ operator is a special case, is smooth if the operator itself is smooth (See Theorem 3 and 6 of~\cite[Sec. 6.3]{evans2022partial}).
Therefore $\psi_j(\cdot)$ are smooth since $\frac{\dee^2}{{\dee x}^2}$ and $\dwfunc$ are smooth.
By analytic methods such as WKB approximation it is straightforward that $\psi_j(\cdot)$ has $C^1$-norm $h^{-\bigO{j}} \leq \log^{\bigO{j}}(G_{\circ})$ for any $j$.
\begin{claim} \label{clm:x-property-pre}
    $C_h \in \bigO{G_\circ^{1+o(1)}} \cap \bigOmega{G_\circ^{1-o(1)}}$, $E_2 - E_1 \geq \bigOmega{G_\circ^{1-o(1)}}$ and $E_1 - E_0 = 2$.
\end{claim}
\begin{proof}
    \thm{double-well-main}(i) and (ii) give $E_1 - E_0 \in \bigOmega{e^{-\frac{S_0 + o(1)}{h}}} \cap \bigO{e^{-\frac{S_0 - o(1)}{h}}}$ and $E_2 - E_1 \geq \Omega(h)$ if $C_h = 1$ in the definition of $\widehat{X}$.
    Now that $C_h$ is chosen such that $E_1 - E_0 = 2$, therefore $C_h \in \bigO{e^{\frac{S_0 + o(1)}{h}}} \cap \bigOmega{e^{\frac{S_0 - o(1)}{h}}}$.
    Hence $E_2 - E_1 \geq \bigOmega{h e^{\frac{S_0 - o(1)}{h}}}$.
    We conclude with the definition of $G_\circ$.
\end{proof}
Now, if we write $G = C_h h^2$, it is immediate that $G \in \bigO{G_\circ^{1+o(1)}} \cap \bigOmega{G_\circ^{1-o(1)}}$ by \clm{x-property-pre}.
This means $h^2 C_h \to +\infty$ when $h \to 0$.
Note that the implicit map $h \mapsto C_h$ is continuous.
Therefore there always exists an $h$ satisfying $G = C_h h^2$ given any sufficiently large $G$.
We can define $\hat{X}(G) \coloneqq \hat{X}(h)$ by choosing such an $h$.
We also write the $(j+1)$th eigenpair of $\widehat{X}(G)$ as $(E_j^G, \ket{\psi_j^G})$ if we want to highlight the value of $G$.
Note that
$$
    \widehat{X}(G) =  -G \frac{\dee^2}{{\dee x}^2} + C_h \dwfunc.
$$
Let us rewrite \clm{x-property-pre} in terms of $G$:
\begin{claim}[Property of $\widehat{X}$] \label{clm:x-property}
    $C_h \in \bigO{G^{1+o(1)}} \cap \bigOmega{G^{1-o(1)}}$, $E_2 - E_1 \geq \bigOmega{G^{1-o(1)}}$ and $E_1 - E_0 = 2$.
\end{claim}

Next, define the diagonal operator (in the position basis) $\widehat{Z} \coloneqq \widehat{Z}(w)$, which depends on a parameter $w >0$, to be
$$
    \widehat{Z}(w) \coloneqq \sgn_{w},\text{\quad i.e.,\quad }(\widehat{Z}\cdot f)(x) \coloneqq \sgn_{w}(x) f(x). 
$$
where $\sgn_{w}: \mathbb{R} \to [-1,1]$ is a fixed smooth $(2/w)$-Lipschitz function such that $\sgn_{w}(x) = \sgn(x)$ for $\abs{x} \geq w$:
\begin{equation*}
    \begin{gathered}
        \sgn_{w}(x) \coloneqq \frac{r\left(\frac{1+x/w}{2}\right) - r\left(\frac{1-x/w}{2}\right)}{r\left(\frac{1+x/w}{2}\right) + r\left(\frac{1-x/w}{2}\right)}, \quad \text{where} \quad r(x) \coloneqq \left\{
            \begin{array}{ll}
                e^{-1/x} & x > 0, \\
                0 & x \leq 0.
            \end{array}
        \right.
    \end{gathered}
\end{equation*}
We know that the Pauli-$Z$ written as a matrix in the eigenbasis of Pauli-$X$ is $\ket{+}\bra{-} + \ket{-}\bra{+}$.
Let us show the similar property for $\widehat{Z} \coloneqq \widehat{Z}(w)$ and $\widehat{X}$:
Denote $P_{-}$ as the projector onto the ``low-energy'' space $\mathcal{T}_{-} \coloneqq \spn(\ket{\psi_0},\ket{\psi_1})$ and $P_{+}$ as the projector onto the ``high-energy'' space $\mathcal{T}_{+} \coloneqq \mathcal{T}_{-}^{\perp} = \spn(\ket{\psi_2},\ket{\psi_3},\dots)$.
We will prove that
$$
    P_{-}\widehat{Z}P_{-} \approx \ket{\psi_0}\bra{\psi_1} + \ket{\psi_1}\bra{\psi_0}, \quad P_{+}\widehat{Z}P_{-} \approx \vect{0}.
$$
\begin{claim}[Property of $\widehat{Z}$] \label{clm:z-property}
    There exists a function $\epsilon(w)$ satisfying $\lim_{w\to 0} \epsilon(w) = 0$ such that the following holds:
    \begin{enumerate}[label=(\roman*)]
        \item $P_{-} \widehat{Z} P_{-} = (1 \pm \bigO{r}) (\ket{\psi_0}\bra{\psi_1} + \ket{\psi_1}\bra{\psi_0})  \pm \bigO{r} \ket{\psi_0}\bra{\psi_0} \pm \bigO{r} \ket{\psi_1}\bra{\psi_1}$,
        \item $\norm{P_{+} \widehat{Z} P_{-}} \leq \bigO{r}$,
        \item $\norm{\widehat{Z}} = 1$,
    \end{enumerate}
    where $r \coloneqq \frac{1}{\sqrt{G^{1-\epsilon(w)}}}$ and $\ket{\psi_0},\ket{\psi_1}$ are eigenstates of $\hat{X}(G)$.
\end{claim}
\begin{proof}
We focus on (i) and (ii) since (iii) is obvious from the definition of $\widehat{Z}$.
To begin, recall from \thm{double-well-main}(iii) that there exists a $\ket{\logiczero} \coloneqq \frac{1}{\sqrt2} \left( \ket{\psi_0} + c \ket{\psi_1} \right)$ where $c \in \{\pm 1\}$ such that it is bounded away from $x < x^\star$.
To be more specific, plug $x_0 \coloneqq w$ into \thm{double-well-main}(iii):
$$
    \int_{-1}^{w} \abs{\logiczero(x)}^2 \dee x \leq \bigO{e^{-(2-o(1))\frac{\agmon{x^\star}{w}}{h}}} = \bigO{\frac{1}{G^{1-\epsilon(w)}}},
$$
where $\agmonsymbol \coloneqq \agmonsymbolandparams{\dwfunc}{0}$.
The last equality above comes from the choice of $h$ in the definition of $\widehat{X}$ such that $e^{\frac{S_0}{h}}=G_\circ$ and $G \in \bigO{G_\circ^{1+o(1)}} \cap \bigOmega{G_\circ^{1-o(1)}}$, and the fact that i) $\agmon{x^\star}{-x^\star} = 2\agmon{x^\star}{0}$ from the symmetry of $\dwfunc$ and ii) $\agmon{x^\star}{w} \to \agmon{x^\star}{0}$ when $w \to 0$ since $\dwfunc(x) > 0$ for any $x > 0$.
Therefore,
\begin{align*}
    \braket{\logiczero|\widehat{Z}|\logiczero} &= \int_{-1}^{1} \sgn_w(x) \logiczero(x)^2 \\
    &= \int_{-1}^{w} \sgn_w(x) \logiczero(x)^2 + \int_{w}^{1} \sgn_w(x) \logiczero(x)^2 \\
    &\geq  \int_{-1}^{w} (-1)\cdot \logiczero(x)^2 + \int_{w}^{1} \sgn_w(x) \logiczero(x)^2 \\
    &= -\bigO{\frac{1}{G^{1-\epsilon(w)}}} + \int_{w}^{1} \logiczero(x)^2 \\
    &\geq 1 - \bigO{\frac{1}{G^{1-\epsilon(w)}}}.
\end{align*}
A similar argument gives us 
$$
    \braket{\logicone|\widehat{Z}|\logicone} \leq -1 + \bigO{\frac{1}{G^{1-\epsilon(w)}}},
$$
where $\ket{\logicone} \coloneqq \frac{1}{\sqrt2} \left( \ket{\psi_0} - c \ket{\psi_1} \right)$.
On the other hand, note that $\norm{\widehat{Z}\ket{\logiczero}} \leq 1$ since $\norm{\widehat{Z}}=1$.
We have
$$
    \braket{\logicone | \widehat{Z} | \logiczero} = \norm{\ket{\logicone}\bra{\logicone}\widehat{Z}\ket{\logiczero}} \leq \norm{(I - \ket{\logiczero}\bra{\logiczero})\widehat{Z}\ket{\logiczero}} \leq \sqrt{1 - \left( 1 - \bigO{\frac{1}{G^{1-\epsilon(w)}}} \right)^2 } = \bigO{\frac{1}{\sqrt{G^{1-\epsilon(w)}}}},
$$
where the last ``$\leq$'' follows from the Pythagorean theorem on the triangle consisting of $\widehat{Z}\ket{\logiczero}$ and its projections $\ket{\logiczero}\bra{\logiczero}\widehat{Z}\ket{\logiczero}$ and $(I-\ket{\logiczero}\bra{\logiczero})\widehat{Z}\ket{\logiczero}$.
The same argument can be used again to obtain the $\bigO{\frac{1}{\sqrt{G^{1-\epsilon(w)}}}}$ upper bound  for $\norm{P_{+}\widehat{Z}\ket{\logiczero}}$ and $\norm{P_{+}\widehat{Z}\ket{\logicone}}$.
Using the estimate above we conclude the proof.
\end{proof}

\paragraph{Construction of $\widehat{H}$.}
Let $G^\star$ be a parameter to be determined at the very end.
Define $G_u \coloneqq G^\star / a_u$ for $u=1,\dots,n$.
For notational convenience also define $G_{\rm min} = \min_{1\leq u\leq n} G_u$.
We have  $G_{\rm min} \geq G_{\star} / M_1$ by the assumption of coefficients $\{a_u\}$.
Recall that
$$
    H = \sum_{1 \leq u \leq n} (a_u X_u + b_u Z_u) + \sum_{1 \leq u,v \leq n} b_{u,v} Z_u Z_v.
$$
We construct the corresponding \sdg\ operator $\widehat{H}$: $H^2(\mathbb{R}^n) \to H^2(\mathbb{R}^n)$ as
$$
    \widehat{H} \coloneqq \widehat{H_X} + \widehat{H_Z},
$$
where
\begin{align*}
    \widehat{H_X} &\coloneqq \sum_{1 \leq u \leq n} a_u \widehat{X_u}(G_u), \\
    \widehat{H_Z} &\coloneqq \sum_{1 \leq u \leq n} b_u \widehat{Z_u} + \sum_{1 \leq u,v \leq n} b_{u,v} \widehat{Z_u}\widehat{Z_v},
\end{align*}
in which $\widehat{X_k}(G_k)\coloneqq I^{\otimes k-1} \otimes \widehat{X}(G_k) \otimes I^{\otimes n-k}$ is $\widehat{X}(G_k)$ acting on the $k$th coordinate;
similarly, $\widehat{Z_k}$ is $\widehat{Z}$ acting on the $k$th coordinate and $\widehat{Z_j}\widehat{Z_k}$ is two $\widehat{Z}$ acting on the $j$th and the $k$th coordinates simultaneously.
Let $P_{-}^G$ be the projector onto the space $\mathcal{T}_{-}^G \coloneqq \spn(\ket{\psi_0^G},\ket{\psi_1^G})$ and $P_{+}^G$ be the projector onto the space $\mathcal{T}_{+}^G \coloneqq (\mathcal{T}_{-}^G)^{\perp} = \spn(\ket{\psi_2^G},\ket{\psi_3^G},\dots)$, where $\{\ket{\psi_k^G}\}_{k \in \mathbb{N}}$ are eigenstates of $\widehat{X}(G)$.
We omit the superscript ``$G$'' when it is clear from the context.
The reason why we can relate $H$ and $\widehat{H}$ is that the ``low-energy'' part of $\widehat{H}$ is basically $H$.
To be more specific, define $\mathcal{S} \coloneqq \bigotimes_{1 \leq k \leq n} \mathcal{T}_{-}^{G_k}$ as the ``low-energy'' subspace:
\begin{claim} \label{clm:main-reduction.hvsh}
    Define unitary $W \coloneqq \bigotimes_{1 \leq k \leq n} (\ket{\psi_0^{G_k}}\bra{-} + \ket{\psi_1^{G_k}}\bra{+})$ that shifts $\{\ket{-},\ket{+}\}$ basis to $\{\ket{\psi_0^{G_k}},\ket{\psi_1^{G_k}}\}$ basis.
    Then,
    $$
        \norm{\widehat{H}_{\mathcal{S}} - W (H + cI) W^\dagger} \leq \bigO{n^2 M_1 r},
    $$
    where $\widehat{H}_{\mathcal{S}}$ is $P_{\mathcal{S}} \widehat{H} P_{\mathcal{S}}$ restricted on $\mathcal{S}$, and $c \coloneqq (E_0+1) \sum_{1 \leq u \leq n} a_u$ and $r \coloneqq \frac{1}{\sqrt{G_{\rm min}^{1-\epsilon(w)}}}$ with $\epsilon(\cdot)$ defined in \clm{z-property}.
\end{claim}
\begin{proof}
    We decompose $H$ as $H = H_X + H_Z$ the same way for the decomposition $\widehat{H} \coloneqq \widehat{H_X} + \widehat{H_Z}$.
    By definition we have $P_{-} \widehat{X} P_{-} = E_0 \ket{\psi_0} \bra{\psi_0} + E_1 \ket{\psi_1} \bra{\psi_1}$.
    Note that $E_1 - E_0 = 2$ by \clm{x-property}.
    Therefore,
    \begin{equation} \label{eq:main-reduction.phxp}
        P_{\mathcal{S}} \widehat{H_X} P_{\mathcal{S}} = W \left(H_X + I (E_0+1) \sum_{1 \leq u \leq n} a_u\right) W^\dagger.
    \end{equation}
    Similarly, by \clm{z-property}(i), we have
    \begin{equation} \label{eq:main-reduction.phzp}
        \norm{P_{\mathcal{S}} \widehat{H_Z} P_{\mathcal{S}} - W H_Z W^\dagger} \leq \bigO{r} \left( \sum_{1 \leq u \leq n} b_u + \sum_{1 \leq u,v \leq n} b_{u,v} \right) \leq \bigO{n^2 M_1 r}.
    \end{equation}
    Combining \eqn{main-reduction.phxp} and \eqn{main-reduction.phzp} we obtain the bound.
\end{proof}

Now our goal is to apply \lem{pert-spec} on $\widehat{H}$ to argue that the ``high-energy'' part of $\widehat{H}$ does not matter (much).
Before that, we need to estimate relevant parameters $R \coloneqq P_{\mathcal{S}^\perp} \widehat{H} P_{\mathcal{S}}$ and $\Delta \coloneqq \mineigen(\widehat{H}_{\mathcal{S}^\perp}) - \maxeigen(\widehat{H}_{\mathcal{S}})$, where $\widehat{H}_{\mathcal{S}^\perp}$ is $P_{\mathcal{S}^\perp} \widehat{H} P_{\mathcal{S}^\perp}$ restricted on $\mathcal{S}^\perp$.
\begin{claim}[Estimate of $\norm{R}$] \label{clm:main-reduction.r-estimate}
    Inherit $r$ from \clm{main-reduction.hvsh}.
    We have $\norm{P_{\mathcal{S}^\perp} \widehat{X}_k P_{\mathcal{S}}}$, $\norm{P_{\mathcal{S}^\perp} \widehat{Z}_k P_{\mathcal{S}}}$ and $\norm{P_{\mathcal{S}^\perp} (\widehat{Z}_j \widehat{Z}_k) P_{\mathcal{S}}}$ are bounded by $\bigO{r}$.
    Consequently, $\norm{R} = \norm{P_{\mathcal{S}^\perp} \widehat{H} P_{\mathcal{S}}} \leq \bigO{n^2 M_1 r}$.
\end{claim}
\begin{proof}
    Without loss of generality assume $k=1$ and $j=2$.
    First, $P_{\mathcal{S}^\perp} \widehat{X}_1 P_{\mathcal{S}} = 0$ because $\mathcal{S}$ is spanned by the tensor of eigenstates of $\widehat{X}$.
    Next we examine $P_{\mathcal{S}^\perp} \widehat{Z}_1 P_{\mathcal{S}}$: \clm{z-property}(ii) implies 
    \begin{equation} \label{eq:main-reduction.pzp1}
        \norm{(P_{+} \otimes I^{\otimes n-1})\widehat{Z_1}P_{\mathcal{S}}} \leq \norm{(P_{+} \otimes I^{\otimes n-1})\widehat{Z_1}(P_{-} \otimes I^{\otimes n-1})} = \bigO{r}. 
    \end{equation}
    Write
    $
        P_{\mathcal{S}^\perp} = P_{+} \otimes I^{\otimes n-1} + P_{-} \otimes (I^{\otimes n-1} - P_{-}^{\otimes n-1})
    $
    and the only thing left to estimate is $(P_{-} \otimes (I^{\otimes n-1} - P_{-}^{\otimes n-1})) \widehat{Z_1} P_{\mathcal{S}}$:
    \begin{align} \label{eq:main-reduction.pzp2}
        (P_{-} \otimes (I^{\otimes n-1} - P_{-}^{\otimes n-1})) \widehat{Z_1} P_{\mathcal{S}} &= (P_{-} \otimes (I^{\otimes n-1} - P_{-}^{\otimes n-1})) (\widehat{Z} \otimes I^{\otimes n-1}) (P_{-} \otimes P_{-}^{\otimes n-1}) \nonumber \\
        &= (P_{-} \widehat{Z} P_{-}) \otimes ((I^{\otimes n-1} - P_{-}^{\otimes n-1})I^{\otimes n-1}P_{-}^{\otimes n-1} ) \nonumber \\
        &=  (P_{-} \widehat{Z} P_{-}) \otimes 0 = 0. 
    \end{align}
    Therefore $\norm{P_{\mathcal{S}^\perp} \widehat{Z}_1 P_{\mathcal{S}}} \leq \bigO{r}$ from \eqn{main-reduction.pzp1} and \ref{eq:main-reduction.pzp2}.
    Last we estimate $P_{\mathcal{S}^\perp} (\widehat{Z}_1 \widehat{Z}_2) P_{\mathcal{S}}$:
    Use \clm{z-property}(ii) again we have
    \begin{equation*}
        \begin{gathered}
            \norm{(P_{-} \otimes P_{-} \otimes I^{\otimes n-2}) (\widehat{Z_1} \widehat{Z_2}) P_{\mathcal{S}}} = \bigO{r^2} \\
            \norm{(P_{-} \otimes P_{+} \otimes I^{\otimes n-2}) (\widehat{Z_1} \widehat{Z_2}) P_{\mathcal{S}}} = \bigO{r} \\
            \norm{(P_{+} \otimes P_{-} \otimes I^{\otimes n-2}) (\widehat{Z_1} \widehat{Z_2}) P_{\mathcal{S}}} = \bigO{r}
        \end{gathered}
    \end{equation*}
    Write $P_{\mathcal{S}^\perp} = (P_{-} \otimes P_{-} + P_{-} \otimes P_{+} + (P_{+} \otimes P_{-}) \otimes I^{\otimes n-2} + P_{-} \otimes P_{-} \otimes  (I^{\otimes n-2} - P_{-}^{\otimes n-2})$.
    Similar to the argument in \eqn{main-reduction.pzp2} we have
    $$
        (P_{-} \otimes P_{-} \otimes  (I^{\otimes n-2} - P_{-}^{\otimes n-2}))(\widehat{Z_1}\widehat{Z_2}) P_{\mathcal{S}} = 0.
    $$
\end{proof}
\begin{claim}[Estimate of $\Delta$] \label{clm:main-reduction.delta-estimate}
    $\Delta \coloneqq \mineigen(\widehat{H}_{\mathcal{S}^\perp}) - \maxeigen(\widehat{H}_{\mathcal{S}}) \geq \bigOmega{\frac{G_{\rm min}^{1-o(1)}}{M_2}}$ if there exists a $\delta > 0$ such that $G_{\rm min} \geq (n^2 M_1 M_2)^{1+\delta}$ always hold.
\end{claim}
\begin{proof}
    Recall that $\mathcal{S} \coloneqq \bigotimes_{1 \leq k \leq n} \mathcal{T}_{-}^{G_k}$ and $\mathcal{T}_{-}^{G} = \spn(\ket{\psi_0^{G}},\ket{\psi_1^G})$.
    Without loss of generality let $a_1$ be the smallest among $a_u$.
    We immediately get 
    $$
        \begin{gathered}
            \mineigen(P_{\mathcal{S}^\perp}\widehat{H_X} P_{\mathcal{S}^\perp}) = a_1 E_2 + (a_2 + \cdots + a_n)E_0,\\
            \maxeigen(P_{\mathcal{S}}\widehat{H_X} P_{\mathcal{S}}) = (a_1 + \cdots + a_n)E_1.
        \end{gathered}
    $$
    Note that $\widehat{H_Z}$ is bounded since $\norm{\widehat{Z}_u} = \norm{\widehat{Z}_u \widehat{Z}_v} = 1$:
    $$
        \norm{\widehat{H_Z}} \leq \sum_{1 \leq u \leq n} \abs{b_u} + \sum_{1 \leq u,v \leq n} \abs{b_{u,v}}.
    $$
    Therefore by Weyl's inequality (\lem{weyl}), 
    \begin{align*}
        &\phantom{=} \ \  \mineigen(P_{\mathcal{S}^\perp} \widehat{H} P_{\mathcal{S}^\perp}) - \maxeigen(P_{\mathcal{S}} \widehat{H} P_{\mathcal{S}})  \\
        & \geq \left( \mineigen(P_{\mathcal{S}^\perp} \widehat{H_X} P_{\mathcal{S}^\perp})- \norm{P_{\mathcal{S}^\perp} \widehat{H_Z} P_{\mathcal{S}^\perp}} \right) - \left(\maxeigen(P_{\mathcal{S}} \widehat{H} P_{\mathcal{S}}) +\norm{P_{\mathcal{S}} \widehat{H_Z} P_{\mathcal{S}}}\right) \\
        & \geq a_1(E_2 - E_1) - (a_2 + \dots + a_n)(E_1 - E_0) - 2 \left( \sum_{1 \leq u \leq n} \abs{b_u} + \sum_{1 \leq u,v \leq n} \abs{b_{u,v}} \right) \\
        & \geq \frac{1}{M_2} (E_2 - E_1) - (n-1)M_1(E_1 - E_0) - 2n(n+1)M_1.
    \end{align*}
    We conclude the proof with $E_1 - E_0 = 2$ and $E_2 - E_1 = \bigOmega{G^{1-o(1)}}$ from \clm{x-property}.
\end{proof}
We are ready to apply \lem{pert-spec} on $\widehat{H}$ to obtain the following claim.

\begin{claim} \label{clm:main-res-1}
    \begin{equation} \label{eq.main-reduction.final}
        \abs{\eigen_k(H) - \eigen_k(\widehat{H}-cI)} \leq \bigO{n^2M_1 r},
    \end{equation}
    where $c \coloneqq (E_0+1) \sum_{1 \leq u \leq n} a_u$ and $r$ is defined in \clm{main-reduction.hvsh}.
\end{claim}
\begin{proof}
\clm{main-reduction.r-estimate} and \ref{clm:main-reduction.delta-estimate} give $\norm{R} \leq \bigO{n^2M_1 r}$ and $\Delta \geq \bigOmega{G_{\rm min}^{1-o(1)}/M_2} > 0$.
Therefore by \lem{pert-spec}, the low-energy spectrum of $\widehat{H}$ matches that of $\widehat{H}_{\mathcal{S}}$ modulo some negligible error.
To be specific,
\begin{equation} \label{eq:main-reduction.eigen-htohs}
    \abs{\eigen_k(\widehat{H}) - \eigen_k(\widehat{H}_{\mathcal{S}})} \leq \norm{R} \leq \bigO{n^2M_1 r}
\end{equation}
for every $k \in \{0,1,\dots,2^n-1\}$.
Recall that \clm{main-reduction.hvsh} says $\norm{\widehat{H}_{\mathcal{S}} - W (H + cI) W^\dagger} \leq \bigO{n^2 M_1 r}$.
Applying Weyl's inequality (\lem{weyl}) it converts to
\begin{equation} \label{eq:main-reduction.eigen-hstoqubit}
    \abs{\eigen_k(\widehat{H}_{\mathcal{S}}) - (\eigen_k(H) + c)} \leq \bigO{n^2M_1 r}.
\end{equation}
Here we used the fact that the spectrum of $W (H + cI) W^\dagger$ is that of $H$ shifted by $c$ since $W$ is unitary.
The claim follows from \eqn{main-reduction.eigen-htohs} and \ref{eq:main-reduction.eigen-hstoqubit}.
\end{proof}

\begin{claim} \label{clm:main-res-2}
    \begin{equation*}
        \norm{e^{-iHt} - W^\dagger P_{\mathcal{S}} e^{-i(\widehat{H}-cI)t}W} \leq \bigO{n^3 (M_1 rt)^{3/2}}.
    \end{equation*}
\end{claim}
\begin{proof}
    \clm{main-reduction.r-estimate} gives $\norm{R} \leq \bigO{n^2M_1 r}$.
    Therefore by \lem{pert-sim} we have
    \begin{equation*}
        \norm{P_{\mathcal{S}}e^{-i\widehat{H}t}P_{\mathcal{S}} - e^{-i\widehat{H}_{\mathcal{S}}t}} \leq \bigO{\left(\norm{R}t\right)^{3/2}} = \bigO{n^3 (M_1 rt)^{3/2}}.
    \end{equation*}
    Plug the bound $\norm{\widehat{H}_{\mathcal{S}} - W (H + cI) W^\dagger} \leq \bigO{n^2 M_1 r}$ from \clm{main-reduction.hvsh} into \lem{pert-sim-diff}:
    \begin{equation*}
        \norm{e^{-i\widehat{H}_{\mathcal{S}}t} - e^{-ict} We^{-iHt}W^\dagger} = \norm{e^{-i\widehat{H}_{\mathcal{S}}t} - e^{-i(W(H+cI)W^\dagger)t}} \leq \bigO{n^2 M_1 r t}.
    \end{equation*}
    
\end{proof}

Note that we can eliminate $c$ in Eq.~\ref{eq.main-reduction.final} by replacing $\widehat{H}$ with $\widehat{H}-c$.
Finally, let us determine $G^\star$:
\begin{itemize}
    \item \clm{main-reduction.delta-estimate} requires $G_{\rm min} \geq (n^2 M_1 M_2)^{1+\delta}$ for some $\delta > 0$.
    \item To conclude the proof of (ii) from Eq.~\ref{eq.main-reduction.final}, we straightforwardly need $n^2M_1 r \ll \epsilon_1$.
    Recall that $r \coloneqq \frac{1}{\sqrt{G_{\rm min}^{1-\epsilon(w)}}}$ from \clm{z-property},
    which can be satisfied if $G_{\rm min} \geq \left( \frac{n^2 M_1}{\epsilon_1} \right)^{2+\delta}$ for some $\delta>0$ depending on $w$.
    \item Similar to the above, we need $n^3 (M_1 r t)^{3/2} \ll \epsilon_1$ to prove (iii), which can be satisfied if $G_{\rm min} \geq \left( \frac{n^4 M_1^2 t^2}{\epsilon_1^{4/3}} \right)^{1+\delta}$.
\end{itemize}
Therefore we can always choose $G^\star$ such that $G^\star \leq \poly(n,M_1,M_2,1/\epsilon_1,t)$ since $G^\star \leq G_{\rm min}M_1$.

\section*{Acknowledgment} 
We thank Mahathi Vempati and Kewen Wu for insightful feedbacks.
We thank Andrew Childs for helpful discussions on several related works and the computational complexity of simulating TIM.
This work was partially supported by the U.S. Department of Energy, Office of Science, Office of Advanced Scientific Computing Research, National Quantum Information Science Research Centers, Quantum Systems Accelerator, Accelerated Research in Quantum Computing under Award Number DE-SC002027, Air Force Office of Scientific Research under award number FA9550-21-1-0209, the U.S. National Science Foundation grant CCF-1955206 and CCF-1942837 (CAREER), a Simons Investigator award (Grant No. 825053), the Simons Quantum Postdoctoral Fellowship, and a Sloan research fellowship.

\bibliographystyle{alpha}
\bibliography{ref}

\end{document}

%% file: preamble.tex
\newcommand{\agmonsymbol}{d}
\newcommand{\agmonsymbolandparams}[2]{
    \agmonsymbol_{#1,#2}
}
\newcommand{\agmonfull}[4]{
    \agmonsymbolandparams{#1}{#2}(#3,#4)
}
\newcommand{\agmon}[2]{
    \agmonsymbol(#1,#2)
}
\newcommand{\ball}[3]{
    B_{#1}(#2,#3)
}
\newcommand{\ballopen}[3]{
    \mathring{B}_{#1}(#2,#3)
}
\newcommand{\dee}{
    \mathrm{d}
}
\newcommand{\logiczero}{
    \widehat{s_0}
}
\newcommand{\logicone}{
    \widehat{s_1}
}
\newcommand{\logicvar}[1]{
    \widehat{s_{#1}}
}
\newcommand{\mineigen}{
    \mu_{\rm min}
}
\newcommand{\maxeigen}{
    \mu_{\rm max}
}
\newcommand{\dwfunc}{
    f_{\rm dw}
}
\newcommand{\eigen}{
    \lambda
}

\newcommand{\sdgsim}{{\hypersetup{pdfborder={0 0 1}, linkcolor=.}\hyperref[def:sdgsim]{\normalfont\textsc{\sdg Sim}}}}

\DeclareMathOperator{\poly}{poly}
\DeclareMathOperator{\polylog}{polylog}
\DeclareMathOperator{\spn}{span}
\DeclareMathOperator{\dom}{dom}

\DeclareMathOperator{\diag}{diag}

\DeclareMathOperator{\sgn}{sgn}